\documentclass{ecai} 
\usepackage{latexsym}
\usepackage{amssymb}
\usepackage{amsmath}
\usepackage{sidecap}
\usepackage{amsthm}
\usepackage{booktabs}
\usepackage{enumitem}
\usepackage{graphicx}
\usepackage{color}
\usepackage[dvipsnames]{xcolor}
\newcommand{\nicespace}{\vspace{15px}}
\usepackage{algorithm}
\usepackage{algorithmic}
\usepackage{bbold}
\usepackage{soul}
\usepackage{float} 
\usepackage{amsfonts}
\usepackage{url}
\usepackage{nicefrac}
\usepackage[breaklinks=true,hidelinks]{hyperref}
\usepackage[sort&compress,nameinlink,noabbrev,capitalize]{cleveref}
\usepackage[utf8]{inputenc}
\usepackage{subcaption}
\usepackage{xcolor}
\usepackage{multirow}
\usepackage{subcaption}
\usepackage{setspace}
\usepackage{etoolbox}
\usepackage{listings}
\usepackage[most]{tcolorbox}
\usepackage{changepage} 

\newtheorem{observation}{Observation}

\newtheorem{corollary}{Corollary}

\newtheorem{theorem}{Theorem}

\newtheorem{remark}{Remark}
\newtheorem{definition}{Definition}
\newtheorem{example}{Example}

\newcommand{\BibTeX}{B\kern-.05em{\sc i\kern-.025em b}\kern-.08em\TeX}

\newcommand{\appsymb}{$\heartsuit$}
\newif\ifarxiv
\arxivfalse
\usepackage{etoolbox}
\newcommand{\appref}[1]{{\appsymb}}

\newcommand{\appendixproof}[2]{%
  \ifarxiv{}#2\else{}\gappto{\appendixProofText}
	{
		\subsection{Proof of \cref{#1}}\label{proof:#1}
		#2
	}
	\fi{}
}
\newcommand{\mypara}[1]{\smallskip\noindent\textbf{#1.}\quad}

\newcommand{\bb}[1]{\texttt{\textbf{\{}}#1\texttt{\textbf{\}}}}

\tcbset{
  colback=gray!5!white, 
  colframe=gray!75!black,
  fonttitle=\bfseries,
  boxrule=0.5pt,
  left=2pt,
  right=2pt,
  top=2pt,
  bottom=2pt,
  listing only,
  listing options={basicstyle=\ttfamily\footnotesize,breaks=true}}

\begin{document}

\begin{frontmatter}

%%% Use this command to specify your submission number. In doubleblind mode, it will be printed on the first page.

\paperid{2707} 

\title{A Dynamic Approach to Collaborative Document Writing}

%%% Authors section

%%% Use this combination of commands to specify all authors of your 
%%% paper. Use \fnms{} and \snm{} to indicate everyone's first names 
%%% and surname. This will help the publisher with indexing the 
%%% proceedings. Please use a reasonable approximation in case your 
%%% name does not neatly split into "first names" and "surname".
%%% Specifying your ORCID digital identifier is optional. 
%%% Use the \thanks{} command to indicate one or more corresponding 
%%% authors and their email address(es). If so desired, you can specify
%%% author contributions using the \footnote{} command.

\author[A]{\fnms{Avital}~\snm{Finanser}\thanks{Email: finanser@post.bgu.ac.il, ORCID: 0009-0001-6019-9731}}
\author[A,B]
{\fnms{Nimrod}~\snm{Talmon}\thanks{Email: nimrodtalmon77@gmail.com, ORCID: 0000-0001-7916-0979}}
\address[A]{Ben-Gurion University}
\address[B]{IOG (Input Output Global)}

\begin{abstract}
We introduce a model for collaborative text aggregation in which an agent community coauthors a document (modeled as an \textit{unordered} collection of paragraphs) using a dynamic mechanism: agents propose paragraphs and vote on those suggested by others. We formalize the setting and explore its realizations, concentrating on voting mechanisms that aggregate votes into a single, dynamic document. We focus on two desiderata: the eventual stability of the process and its expected social welfare. Following an impossibility result, we describe several aggregation methods and report on agent-based simulations that utilize natural language processing (NLP) and large-language models (LLMs) to model agents. Using these simulations, we demonstrate promising results regarding the possibility of rapid convergence to a high social welfare collaborative text.
\end{abstract}

\end{frontmatter}

%%%%%%%%%%%%%%%%%%%%%%%%%%%%%%%%%%%%%%%%%%%%%%%%%%%%%%%%%%%%%%%%%%%%%%%%

\section{Introduction}

We are interested in enabling communities to collaboratively write documents. This is valuable, e.g., for building residents to draft their shared rules; for a company/society to co-create their mission statement; or for a DAO to establish its digital constitution~\cite{talmon2023social} (possibly with token-weighted voting). 
Existing platforms like Wikipedia,\footnote{https://en.wikipedia.org} Google Docs,\footnote{https://www.google.com/docs} Notion,\footnote{https://www.notion.so} and Overleaf\footnote{https://overleaf.com} are constrained by governance models that fall into two extremes: \textbf{autocracy}, where a single owner controls the process (and other agents merely provide comments), and \textbf{anarchy}, where unrestricted access may lead to disorganization. These limitations hinder the potential for a collaborative framework for inclusive, \textbf{democratic} text creation.

Here, we propose a general approach for collaborative text creation, with particular focus on collaborative \textbf{constitution writing}. We formally describe and analyze the model from a social choice perspective -- both theoretically and via computer-based simulations (including using LLMs). Informally, the approach is as follows:\footnote{Our inspiration is from Consenz~\cite{consenz2023}, a (currently-funded) blockchain-oriented platform for text aggregation that was operated within the Cardano ecosystem. In a way, our work formalizes a stylized and generalized version of its mechanism.}
\begin{itemize}
    
\item
% \textbf{Dynamic process}: The process is dynamic, so that community members can -- at each point in time -- propose new paragraphs to be included in the joint document, as well as vote on paragraphs suggested by other agents.
\textbf{Dynamic process}: The process is dynamic, allowing community members to propose new paragraphs and vote on paragraphs suggested by other agents at any point in time.

\item 
\textbf{Continuous updates}: An aggregated text document is continuously updated as the result of agent operations (propose and vote).

\end{itemize}

From a social choice perspective, the core challenge is to aggregate voters’ preferences (as revealed by their sequence of dynamic operations) into the evolving collective crafted text.
Informally, and corresponding to the use-case collaborative constitution writing, we aim to identify aggregation rules that achieve two key objectives: \textbf{stability} (ensuring that, over time, the aggregated text undergoes fewer changes and becomes less and less volatile -- this is especially important for text that should be solid, like constitutions) and \textbf{high social welfare} (ensuring that, at convergence, the aggregated text reflects a result that voters find satisfying and aligns with their preferences). \footnote{Our definition of high social welfare is inherently subjective; alternatively, epistemic approaches may be considered as well.}

After reviewing related work and platforms and formalizing our setting (Section~\ref{section:formal model}), we go on to our theoretical analysis: first, we formalize axiomatic properties beneficial for reasoning about the different possible aggregation methods -- in particular, we consider social welfare maximization and stability in our context (Section~\ref{section:social welfare and stability}). To principally consider our rule design space, we then describe a rich family of aggregation rules, termed Consensus Conditioned Rules (CCR; Section~\ref{section:CCR}) -- these rules assign scores to individual paragraphs and include in the aggregated text those that score sufficiently high; crucially, the scoring becomes more stringent as the process progresses. We then show that even though social welfare maximization and stability are indeed mutually exclusive (Theorem~\ref{corollary:mutually exclusive}), some CCRs strike good balances between these two desiderata -- concretely, we report simulation-based results (Section~\ref{section:simulations}) in which we simulate different agent populations as they iteratively interact with the system. Specifically, we utilize LLMs and show promising results regarding convergence pace and the output quality at convergence.

Proofs (partially) missing from the main text are denoted by $\heartsuit$.

\mypara{Related Work}%\label{section:related work}
%
% We discuss relevant literature.
%
% \paragraph{Collaborative Platforms and Democratic Systems.} 
%
The landscape of \textit{collaborative/democratic platforms} has evolved significantly~\cite{collabDocs}. We mention \emph{LiquidFeedback} -- who pioneered the integration of liquid democracy principles into deliberation-based on decision making~\cite{liquidfeedback}; and \emph{Polis} -- who combines machine learning with collective intelligence to facilitate large-scale democratic deliberation~\cite{polis}. However, these platforms focus on deliberation and decision-making rather than document creation.
Regarding vote elicitation, agents in our model either approve or disapprove each paragraph; this is studied both theoretically~\cite{gonzalez2019dilemma, laruelle2021not} and implemented on certain platforms, such as \emph{Reddit}.\footnote{https://www.reddit.com}

\textit{Theoretically}, as our aggregated texts do not consider paragraph ordering, we relate to multiwinner elections~\cite{mwchapter}, and elections with a variable number of winners~\cite{ApprovalVNW, ComplexityVNW}. In a sense, we provide a dynamic extension of the work of Halpern et al.~\cite{RepresentationIncomplete} -- we introduce dynamic properties where both the candidate pool and the solution size remain unrestricted and continuously evolving.
Filtser et al.~\cite{filtser2019distributed} address the fundamental challenge of maintaining collective decisions as preferences evolve over time, providing algorithms that allow for the maintenance of an aggregation result.  
We also mention the work of Bulteau et al.~\cite{bulteau2021aggregation}, who suggested embedding voter preferences in a metric space and then aggregating them with
general-purpose metric aggregation rules. Most notably, our work differs in introducing a dynamic, iterative context for the aggregation process.

%
% \paragraph{LLMs in Social Choice.} 
%

As for our use of \textit{Large Language Models (LLMs)}, indeed, recent advances in LLMs have spawned interest also in social choice. One relevant line of work is \emph{generative social choice}~\cite{GenerativeSocialChoice}, in which survey-based preference collection and LLM-driven statement generation are combined to capture collective opinions. Building on this, Yang et al.~\cite{LlmVoting} utilized an agent-based modeling approach using LLMs to simulate collective decision-making processes. Our textual context and focus on document creation distinguish our work from these. Note that, while generative social choice utilizes LLMs for algorithmic purposes, we utilize them for searching and validating the solution space via agent-based simulation.

% \avitaaal{This is the edition:} Elevating AI in social choice can be as a solution proposition (the way conducted in general social choice), or as a tool in validating and searching solution space, as we are doing in this work (agent modeling).

%%%%%%%%%%%%%%%%%%%%%%%%%%%%%%%%%%%%%%%%%%%%%%%%%%%%%%%%%%%%%%%%%%%%%%%%

\section{Formal Model}\label{section:formal model}

Our formal model describes the collaborative document writing system. Formally, our setting is as follows (see Figure~\ref{fig:platform} for an illustration of a possible UX for the corresponding system):
\begin{itemize}

\item An infinite set $P = \{p_1, \ldots\}$ of possible paragraphs for inclusion in the aggregated text (i.e., the solution) is given. (Practically, the set is finite and defined by the proposed paragraphs.)

\item
We assume a population of $n$ \emph{agents}, $A = \{a_1, \ldots, a_n\}$.

\item
Agents sequentially generate \emph{events}, representing creation and preferences regarding the paragraphs of $P$. 
%An event \avitaaal{is a tuple} $e = (a, p, v)$, generated by agent $a \in A$ regarding either approving ($v = 1$) or disapproving ($v = -1$) paragraph $p \in P$: so, $(a, p, 1)$ ($(a, p, -1)$) corresponds to $a$ approving (resp., disapproving) $p$. 
An event is a tuple $e = (a, p, v)$
where $a \in A$, $p \in P$, and $v \in \{+1, -1\}$ indicates an approval ($v=+1$) or disapproval ($v=-1$) of paragraph $p$ by agent $a$. Hence, $(a, p, +1)$ ($(a, p, -1)$) corresponds to $a$ approving (resp., disapproving) $p$.
For $e = (a, p, v)$, we denote $a := A(e)$, $p := P(e)$, and $v := V(e)$.
An event $e = (a, p, 0)$ corresponds to agent $a$ withdrawing their vote on $p$ (abstaining).
% (i.e., abstaining) - to save the additional line 
\item 
We define an \emph{event list} $E = \langle e_1,\ldots, e_t\rangle$ as a sequence of (ordered) events.
Correspondingly, a collaborative document writing instance is a pair $\mathcal{T} = (A, E)$ where an agent community $A$ generates an event list~$E$.

\item
Given an instance $\mathcal{T} = (A, E)$, an \emph{aggregation rule} $\mathcal{R}$ produces an aggregated set of paragraphs (i.e., a solution; or, the aggregated text) denoted by $\mathcal{R}(\mathcal{T}) \subseteq P$. It will be useful to address a solution as a binary vector $S = \langle s_1,...,s_{|P|} \rangle$, with $s_i \in \{+1,-1\}$, $i \in [P]$ (where $s_i$ signifies whether the $i$th paragraph is included or excluded from the document $\mathcal{R}(\mathcal{T})$).

\end{itemize}

\begin{remark}
Observe that our aggregated text is an \textbf{unordered set} of paragraphs -- i.e., the order of paragraphs does not matter.
This aligns well with our use cases of collaboratively drafting a set of rules, co-creating a mission statement, or establishing an off-chain constitution for a DAO. We discuss this issue further in Section~\ref{section:outlook}.
\end{remark}

% \subsection{Additional Notation and Remarks} 

We also use the following additional notation:

\begin{itemize}

\item
For an event list $E$, $P(E) := \{ p \in P \,:\, |E\!\downarrow_p| > 0 \}$ denotes the paragraphs considered in $E$.

\item
For an event list $E$, $E\!\downarrow_p:= \langle e \in E: P(e)=p \rangle$ denotes the sub-list of events of $E$ that consider paragraph $p$; similarly, $E\!\downarrow_{a,p}:= \langle e \in E: A(e) = a, P(e) = p \rangle$ denotes the sub-list of events of $E$ generated by agent $a$ regarding paragraph $p$.

\item
For an event list $E = \langle e_1, \ldots, e_t \rangle$, we use $E + e :=  \langle e_1, \ldots, e_t, e \rangle$ to denote the appending of event $e$ to $E$.

\item 
As votes can change the stance of each agent $a$ on paragraph $p$, it is useful to define the following.

\begin{definition}[Stance]
    \label{Stance}
    For an instance $\mathcal{T}= (A, E)$; its \emph{stance} -- $\mathrm{stance}(\mathcal{T})$ -- is a $|A| \times |P|$ matrix, where each entry $\mathrm{stance}(\mathcal{T})_{a,p}\in \{+1, 0, -1\}$ 
    %being 
    is the last recorded 
    %stance
    vote of agent $a$ on paragraph $p$. Formally, 
    $\mathrm{stance}(\mathcal{T})_{a, p} = V(E\!\downarrow_{a,p}[-1])$, where $V(E\!\downarrow_{a,p}[-1])$ is the value
    of $a$’s last event on $p$.
    %the last vote performed by agent $a$ on paragraph~$p$.
\end{definition}

% So, the stance matrix reflects the last recorded stance of all agents in all paragraphs.

\begin{example}\label{ex:stance}
Let $A = \{a_1, a_2, a_3\}$ be an agent community and consider 
$E = \langle (a_1, p_2, +1), (a_2, p_3, +1), (a_1, p_2, 0) \rangle$ and $E' = \langle (a_3, p_1, +1), (a_2, p_3, +1), (a_3, p_1, 0) \rangle$. Note that (1) $P(E) = \{p_2,p_3\}$ and $ P(E') = \{p_1,p_3\}$; and (2) $\mathrm{stance}((A, E)) = \mathrm{stance}((A, E'))$ with $\mathrm{stance}((A, E))_{a, p} = +1$ for $a = a_2$ and $p = p_3$ and $0$ otherwise.  
\end{example}

\item
For an event list $E$ and a paragraph $p$, we 
%denotes
denote $p^+$ (resp., $p^-$) as the number of agents whose final vote for paragraph $p$ is an approval (resp., disapproval); i.e., $p^+:= \sum_{a=1}^{n} \mathbb{1}_{\{\mathrm{stance}(\mathcal{T})_{a,p} = +1\}}$ and $p^-:= \sum_{a=1}^{n} \mathbb{1}_{\{\mathrm{stance}(\mathcal{T})_{a,p} = -1\}}$.

\item
Given the definition of a voter's stance on a paragraph, we can define the \emph{tally} of an instance.
\begin{definition}[Tally]\label{tally}
  The \emph{tally} of an instance $\mathcal{T} = (A, E)$ is an $\{-1, +1\} \times |P|$ matrix, denoted by $\mathrm{tally}(\mathcal{T})$, with $\mathrm{tally}(\mathcal{T})_{v,p}$ being the number of agents whose final vote is typed $v$ for paragraph~$p$. Formally,
    $\mathrm{tally}(\mathcal{T})_{v,p} := \sum_{a \in A} \mathbb{1}_{\{\mathrm{stance}(\mathcal{T})_{a,p} = v\}}$.
\end{definition}

\end{itemize}

\begin{remark}
An event $(a, p, +1)$ does not preclude a subsequent event $(a, p, -1)$, and vice versa. Also, suggesting a new paragraph inherently requires an initial approval event. Thus, for each new paragraph $p \in P$ suggested in event list $E$, the first event satisfies $V(E\!\downarrow_p[1]) = +1$ (i.e., $V(E\!\downarrow_p[1]) = +1$ holds for each $p\in P(E)$).
%
% Old version 
%
% An event $(a, p, 1)$ does not exclude the possibility of a future event $(a, p, -1)$, and vice versa. Moreover, the suggestion of a new paragraph is its first approval event. A new suggested paragraph $p \in P$ in event list $E$ is thus represented as $V(E\!\downarrow_p[1]) = 1$; such that $V(E\!\downarrow_p[1]) = 1$ holds for each $p\in P(E)$.
%
\end{remark}

\begin{remark}
Figure~\ref{fig:platform} illustrates a possible UX corresponding to our model: agents interact with the system using the left \textit{Editor} pane, in which they suggest new paragraphs and vote on existing ones; and the right \textit{Document} pane presents the updated aggregated text -- the result of applying $\mathcal{R}$ on the event list $E$.
\end{remark}

% In particular, a group of axioms that are essential as properties of the aggregation rule are Anonymity, Versatility, and Relevance. In the context of this current work, the rules examined will necessarily satisfy these axioms.

% Anonymity ensures that the aggregation outcome is independent of individual agent identities (as the result is not affected by renaming/permuting agent identities).

%

\section{Social Welfare and Stability}\label{section:social welfare and stability}

As our focal properties, we concentrate on stability (i.e., convergence of the solution) and social welfare maximization (i.e., having the solution reflect the voters' preferences well).
We begin by formalizing social welfare maximization, utilizing a natural satisfaction-based approach. 
Let the number of paragraphs that agent $a$ voted on be $N_a:= \sum\limits_{p=1}^{k} \mathbb{1}_{\{ \mathrm{stance}(\mathcal{T})_{a,p} \neq 0\}}$ (recall that agents need not vote on all paragraphs); then:

\begin{definition}[Agent Satisfaction]
Given an instance $\mathcal{T}$, an agent $a$, a paragraph $p$, and a solution $S$, the \emph{satisfaction of $a$ from $p$ in $S$} is:
\[
    Sat_{a,p}(\mathcal{T}, S) =  
    \begin{cases} 
    \frac{1}{N_a}        & \text{if $s_p = \mathrm{stance}(\mathcal{T})_{a,p}$} \\
    0             & \text{otherwise}
    \end{cases}
\]
The \emph{satisfaction of $a$ from $S$} is $Sat_a(\mathcal{T}, S) = \sum\limits_{p=1}^{k} Sat_{a,p}(\mathcal{T}, S)$.
\end{definition}

\begin{figure}[tb]
    \centering    
    \includegraphics[width=0.48\textwidth]{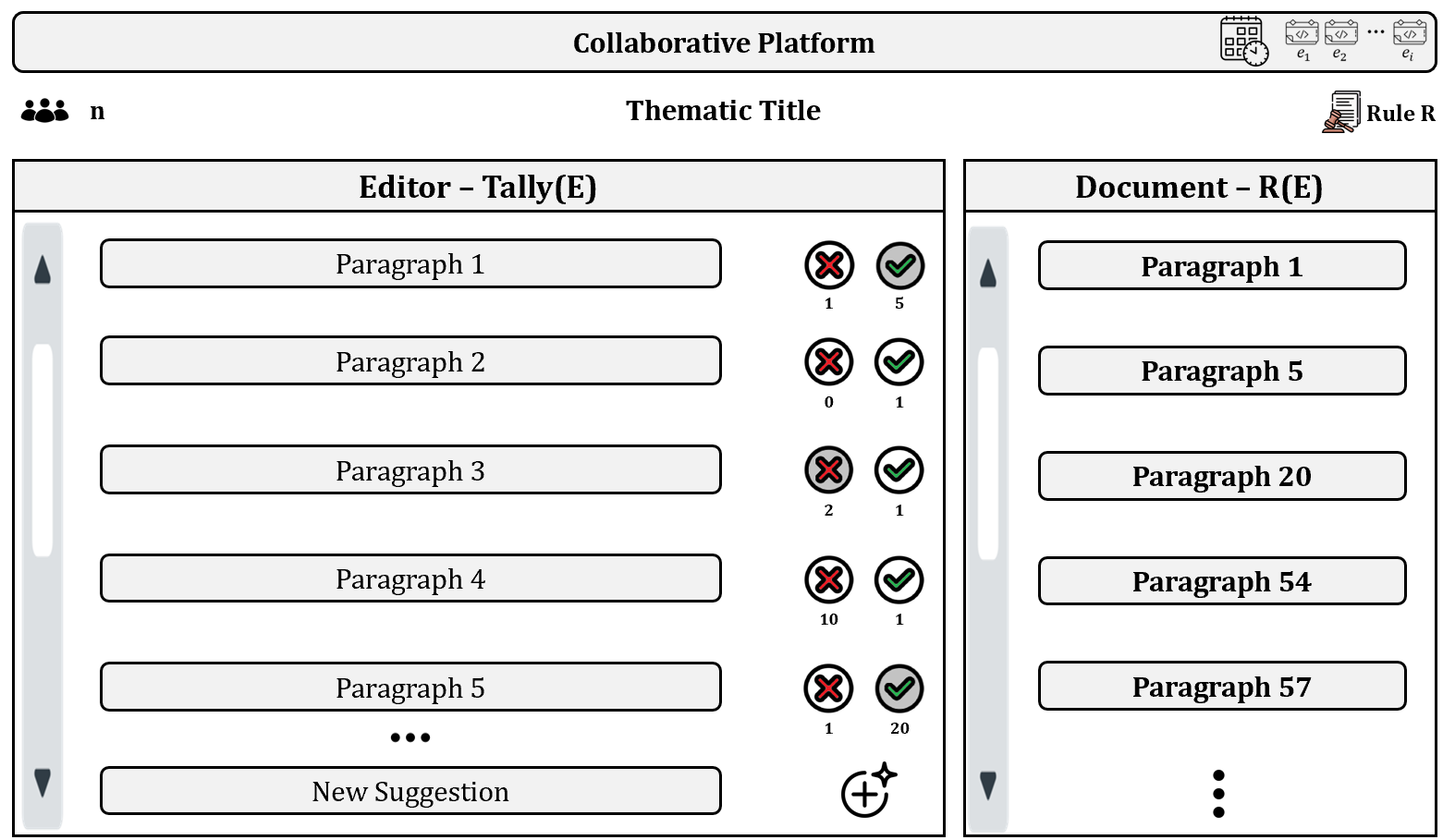} 
    \caption{Collaborative platform for drafting a document. On the left is the editor pane seen by individual agents, enabling them to suggest new paragraphs and vote on existing ones, while on the right is the dynamic aggregated text that is the result of applying the voting rule on agent operations.}
    \label{fig:platform}
    \nicespace
\end{figure}

That is, agent satisfaction from a solution is the \textbf{percentage} of paragraphs for which the solution aligns with her preferences (i.e., she approves and the paragraph gets included; or she disapproves and the paragraph gets excluded) out of all paragraphs she voted on. This measure is reminiscent of the classical notion of \emph{accuracy}; furthermore, the normalization by $N_a$ ensures a \textbf{one-person, one-vote} behavior.
Building on the definition above, the overall satisfaction of the community with a solution is defined as the sum over the agents:

\begin{definition}[Community Satisfaction]
\label{Solution_Satisfaction}
  Let $\mathcal{T}=(A, E)$ be an instance and $S$ a solution. The \emph{satisfaction} of $A$ from $S$ is:
  \[
    Sat_A(\mathcal{T}, S) = \sum\limits_{a=1}^{n} Sat_a(\mathcal{T}, S)= \sum\limits_{a=1}^{n} \sum\limits_{p=1}^{k} \frac{\mathbb{1}_{\{\mathrm{stance}(\mathcal{T})_{a,p} = s_p\}}}{N_a}\ .    
  \]
\end{definition}

A rule is \textit{SWM} if it always maximizes the social welfare.
% }We can define \textit{SWM}: the following axiom -- a rule $\mathcal{R}$ is \emph{maximizing social welfare} if it always outputs a solution with the maximum possible community satisfaction. 

\begin{definition}[Social Welfare Maximizer (SWM)]
\label{SWM}
    An aggregation rule $\mathcal{R}$ satisfies SWM if, for any instance $\mathcal{T}=(A, E)$, the solution $\mathcal{R}(\mathcal{T}) \in \text{argmax}_{S \in 2^{|P|}} Sat_{A}(\mathcal{T}, S)$.
\end{definition}

\subsection{An Efficient Social Welfare Maximizer Rule}

We are interested in identifying social welfare maximizers.
% Next, we characterize these rules. 
% class of such voting rules (these operate by including paragraphs with majority support). First, we introduce some notation.

\begin{definition}
  Given an instance $\mathcal{T} = (A,E)$ and a paragraph $p$, let $p^{+}_{r}$ ($p^{-}_{r}$) denote the \textbf{relative approval} (\textbf{relative disapproval}):
    \[
    p^{+}_{r} := \sum_{a \in A} \frac{\mathbb{1}_{\{\mathrm{stance}(\mathcal{T})_{a,p} = +1\}}}{N_a} \text{ and }    
    p^{-}_{r} := \sum_{a \in A} \frac{\mathbb{1}_{\{\mathrm{stance}(\mathcal{T})_{a,p} = -1\}}}{N_a}\ .
    \]  
\end{definition}

\begin{definition}[Relative Majority Rule]
\label{RM}
  A rule $\mathcal{R}$ is a \emph{relative majority} rule (an \emph{RM rule}) if, for each instance $\mathcal{T}$ and paragraph $p$, it satisfies: (1) $\text{if } p^{+}_{r} > p^{-}_{r} \text{ then } p \in \mathcal{R}(\mathcal{T})$; and (2) $\text{if } p^{+}_{r} < p^{-}_{r} \text{ then } p \notin \mathcal{R}(\mathcal{T})$. (Breaking ties arbitrarily for paragraphs with $p^+_r = p^-_r$.)
  % \end{align*}
\end{definition}

Intuitively, an RM rule includes a paragraph if the proportion of agents approving it exceeds the proportion disapproving it; otherwise, it is excluded, with ties resolved arbitrarily. Below is the characterization result.

\begin{theorem}[\appref{thm:RMisSWM}]\label{thm:RMisSWM}
  A rule satisfies SWM iff it is an RM rule.
\end{theorem}

\appendixproof{thm:RMisSWM}{
\begin{proof}
We prove the two implications separately:
%.

\paragraph{Sufficiency.}
We show that a relative majority rule satisfies SWM. Let $\mathcal{R}_m$ be a relative majority (RM) rule that, for any collaborative document instance $\mathcal{T}=(A, E)$, produces a solution $S_{\mathcal{R}_m} \in \{+1,-1\}^{k}$, where $k=|P(E)|$ (the number of proposed paragraphs in $E$). 
%The accuracy satisfaction 
% Note - we did not call it that way
The community satisfaction score for $S_{\mathcal{R}_m}$ is $Sat(\mathcal{T},S_{\mathcal{R}_m}) =  \sum\limits_{a=1}^{n} \sum\limits_{p=1}^{k} \frac{\mathbb{1}_{\{S_{\mathcal{R}_m}(p) = \mathrm{stance}(\mathcal{T})_{a,p}\}}}{N_a} $ (Definition~\ref{Solution_Satisfaction}). Note that, by commutativity of finite sums, we can decompose this measure as $Sat(\mathcal{T}, S_{\mathcal{R}_m}) =  \sum\limits_{a=1}^{n} \left( \sum\limits_{p=1}^{k} \frac{\mathbb{1}_{\{S_{\mathcal{R}_m}(p) = \mathrm{stance}(\mathcal{T})_{a,p}\}}}{N_a} \right) = \sum\limits_{p=1}^{k} \left( \sum\limits_{a=1}^{n} \frac{\mathbb{1}_{\{S_{\mathcal{R}_m}(p) = \mathrm{stance}(\mathcal{T})_{a,p}\}}}{N_a} \right)$. 
Suppose, for contradiction, that $\mathcal{R}_m$ does not satisfy SWM. Then, there exists a solution $S' \neq S_{\mathcal{R}_m}$ such that $Sat(\mathcal{T}, S_{\mathcal{R}_m}) < Sat(\mathcal{T}, S')$. Let $P'$ denote the set of paragraphs where $S'$ and $S_{\mathcal{R}_m}$ differ. For any $p \in P'$, we have two cases:
\begin{enumerate}
    \item $S_{\mathcal{R}_m}(p) = +1$ and $S'(p)=-1$. By Definition \ref{RM}, as $p \in \mathcal{R}_m(\mathcal{T})$, $p^{+}_{r} > p^{-}_{r}$ which implies that $ \sum\limits_{a=1}^{n} \frac{\mathbb{1}_{\{\mathrm{stance}(\mathcal{T})_{a,p} = +1\}}}{N_a} > \sum\limits_{a=1}^{n} \frac{\mathbb{1}_{\{\mathrm{stance}(\mathcal{T})_{a,p} = -1\}}}{N_a}$. Therefore, $Sat_p(\mathcal{T}, S_{\mathcal{R}_m}(p)= +1) > Sat_p(\mathcal{T}, S'(p)=-1)$.
    \item $S_{\mathcal{R}_m}(p) = -1$ and $S'(p) = +1$. By Definition~\ref{RM}, as $p \notin \mathcal{R}_m(\mathcal{T})$, $p^{+}_{r} < p^{-}_{r}$, similarly to case 1, leading to $Sat_p(\mathcal{T}, S_{\mathcal{R}_m}(p)=-1) > Sat_p(\mathcal{T}, S'(p)=+1)$.
\end{enumerate}

In both cases, we have $Sat_p(\mathcal{T}, S_{\mathcal{R}_m}) > Sat_p(\mathcal{T}, S')$ for all $p \in P'$. Since these are the only paragraphs where $S_{\mathcal{R}_m}$ and $S'$ differ, and $Sat_p(\mathcal{T}, S_{\mathcal{R}_m}) = Sat_p(\mathcal{T}, S')$ for all $p \notin P'$, by our decomposition of $Sat(\mathcal{T}, S)$ into paragraph-wise sums: 
$ Sat(\mathcal{T}, S_{\mathcal{R}_m}) = \sum\limits_{p=1}^{k} Sat_p(\mathcal{T}, S_{\mathcal{R}_m}) > \sum\limits_{p=1}^{k} Sat_p(\mathcal{T}, S') = Sat(\mathcal{T}, S')$. This contradicts our assumption about the existence of such a solution $S'$; therefore, $\mathcal{R}_m$ satisfies SWM.

\paragraph{Necessity.}
We prove that if an aggregation rule $\mathcal{R}$ satisfies SWM, then it must be a relative majority rule. 

Let $\mathcal{R}$ be an aggregation rule that satisfies SWM, and we assume it is not a relative majority rule. Also, consider a relative majority rule $\mathcal{R}_m$ that produces an identical solution to $\mathcal{R}$ except for paragraph $p$ given an instance $\mathcal{T}$: ${\mathcal{R}_m(\mathcal{T})} \triangle {\mathcal{R}(\mathcal{T})} = p$. 

% $S_{\mathcal{R}_m} \triangle S_{\mathcal{R}} = p$. 

As $\mathcal{R}$ and $\mathcal{R}_m$ differ in $p$ and $\mathcal{R}$ not being a relative majority rule, this means either that $p^{+}_{r} > p^{-}_{r}$ but $S_{\mathcal{R}}(p) = -1$, or $p^{+}_{r} < p^{-}_{r}$ but $S_{\mathcal{R}}(p) = +1$. In both cases, the previous section of this proof (sufficiency) showed that for all collaborative document instances $\mathcal{T} \text{ then }
\sum\limits_{a=1}^{n} \frac{\mathbb{1}_{\{S_{\mathcal{R}_m}(p) = \mathrm{stance}(\mathcal{T})_{a,p}\}}}{N_a}  > \sum\limits_{a=1}^{n} \frac{\mathbb{1}_{\{S_{\mathcal{R}}(p) = \mathrm{stance}(\mathcal{T})_{a,p}\}}}{N_a}$, while all other paragraphs have equal satisfaction.
%accuracy
Hence, $Sat(\mathcal{T}, S_{\mathcal{R}}) < Sat(\mathcal{T}, S_{\mathcal{R}_m})$ which contradicts $\mathcal{R}$ satisfying SWM.
\end{proof}
}

\subsection{Stability and Impossibility}

Recall that we are interested in rules that are not only social welfare maximizers but also stable. Next, we formalize this.

\begin{definition}[Stability]
\label{Stability}
    A rule $\mathcal{R}$ is \emph{stable} if, for every event list $E$, there exists an extension $E'$ such that for any further extension $E''$ of $E + E'$,  $\mathcal{R}(E + E') = \mathcal{R}(E + E' + E'')$.
\end{definition}

That is, for any event list $E$, there are further events $E'$ that cause the system to stabilize (such that any further events $E''$ will not change the output).
Unfortunately, no RM rule is stable.

\begin{observation}
[\appref{thm:rm not stable}]\label{thm:rm not stable}
  No RM rule is stable.
\end{observation}

\appendixproof{thm:rm not stable}{
\begin{proof}
Let $\mathcal{R}$ be an RM rule; let $E$ and $E'$ be event list; and $p$ a paragraph in $P(E + E')$. If $p \in \mathcal{R}(E + E')$,  define $E''$ as an event list containing $(p^+ + 1)$ events of the form $(a, p, -1)$.
%such that -- since
Since $\mathcal{R}$ is an RM rule, 
%-- 
after appending $E''$ we have $p \notin \mathcal{R}(E + E' + E'')$, thus violating stability. (Similarly, if $p \notin \mathcal{R}(E + E')$, then add $(p^- + 1)$ down-vote events for $p$ to reach a contradiction).
\end{proof}
}

Theorem~\ref{thm:RMisSWM} and Observation~\ref{thm:rm not stable} imply the following.

\begin{corollary}\label{corollary:mutually exclusive}
  No rule is both SWM and Stable.
\end{corollary}

\section{Consensus-Conditioned Rules}\label{section:CCR}

Following the impossibility result proven above, next we define a rich family of rules, termed Consensus-Conditioned Rules (CCRs); this will enable the establishment of a structured design space\footnote{Taking inspiration from similar approaches for multiwinner rules~\cite{Axiomatic_Classification_and_Hierarchy}.} and the identification of rules that strike a good balance between social welfare maximization and stability.
A CCR corresponds to a \textit{consensus condition}, so paragraphs that satisfy the condition are included in the solution: a score is being computed for each paragraph, and paragraphs whose score reaches some predefined threshold are included. Formally (see also~\autoref{alg:CCR_algorithm} in Appendix~\ref{appendix:CRR_pseudo}):

\begin{definition}[Consensus Scoring Function -- CSF]\label{CSF}
  Let $\mathcal{T}=(A, E)$ be an instance. A \emph{consensus scoring function} is a function $\mathrm{CSF} : (p,\mathcal{T}) \to [0, 1]$.
\end{definition}

% Although infinitely many CSFs are possible, each primary function examined here represents a different perspective on measuring consensus.

Given some CSF, which assigns a score to each paragraph, a threshold $x \in [0, 1]$ is used to define a CCR that includes those paragraphs that pass the threshold.

% \begin{definition}[Consensus Condition -- CC]\label{CC}
%   Let $\text{CSF}$ be a consensus scoring function and $x \in [0,1]$ be a predefined threshold. Given a system instance $\mathcal{T}$, we define the pair $\bigl(\text{CSF}, x\bigr)$ as consensus condition:
%   \[
%     \forall p \in P: \bigl(\text{CSF}(p,\mathcal{T}), x\bigr) \text{ is True } \Leftrightarrow \text{CSF}(p, \mathcal{T}) \geq x
%   \]
% \end{definition}

% \paragraph{Consensus Condition (CC).} A joint pair of a CSF and a threshold $x$ ranging from 0 to 1 -- $\bigl(\text{CSF}, x\bigr)$ define a \emph{consensus condition}. This condition indicates whether a suggested paragraph $p$ has reached a sufficient level of consensus (x) for inclusion in the solution based on a consensus measure metric (CSF).

\begin{definition}[Consensus-Conditioned Rule -- CCR]\label{CCR}
  Given an instance $\mathcal{T} = (A,E)$, a CSF and a threshold $x \in [0, 1]$, the CCR rule $\mathrm{CCR}_{[\mathrm{CSF}, x]}$ satisfies:
    $p \in \mathrm{CCR}_{[\mathrm{CSF}, x]}(\mathcal{T})
    % E)
    \text{ iff } \mathrm{CSF}(p, \mathcal{T}) \geq x$.
\end{definition}

\subsection{Static CCRs}

Next, we articulate a subfamily of CCRs, referred to as \emph{static CCRs}. These include all RM rules and will be useful for our understanding of stability.
Essentially, static CCRs correspond to CSF that operate only on the stance matrix (and not on the complete event list, so that only the last stance of each voter on each paragraph matters).\footnote{Following the logic of C1, C2, and C3 rules~\cite{CondorcetSCF,HandbookCOSMOC}, operating on different input refinements.}

\begin{definition}[static CSF and CCR]\label{staticCSF}
  A \emph{static-CSF} is a CSF for which $\mathrm{CSF}(p, \mathcal{T}) = \mathrm{CSF}(p, \mathcal{T'})$ for any $\mathcal{T}$ and $\mathcal{T'}$ such that $\mathrm{stance}(\mathcal{T}) = \mathrm{stance}(\mathcal{T'})$.
  Consequently, a \emph{static-CCR} is a CCR corresponding to a static CSF.\footnote{Indeed, since the stance matrix is invariant to event order permutations, so are static CCRs.} (If a CSF use only $E$ then we sometimes denote $\mathrm{CSF}(p, E)$ and not $\mathrm{CSF}(p, \mathcal{T})$.)
  % Let $\mathcal{T} =(A, E)$ be a document writing instance, and $\mathrm{stance}(\mathcal{T})$ its stance. A \textbf{static-CSF}, denoted as $C_s$, is a CSF function such that $\forall p \in P\text{, } C_s : (p,\{-1, 0, 1\}^{n \times m}) \to [0, 1]$.
\end{definition}

% \begin{remark}
% %
% A quick observation is that the stance matrix is unaffected by permutations of the order of events in a given event list; similarly, static CCRs are invariant to such permutations.
% %
% \end{remark}

\paragraph{Static CSFs.} We provide some natural static CSFs:
\begin{description}

\item
\textbf{Approval Proportion Score (APS)}:
\begin{equation}\label{eq:aps}
\text{APS}(p, E) = \frac{p^+}{p^+ + p^-}
\end{equation}
I.e., APS computes the percentage of upvotes.

\item
\textbf{Relative Approval Proportion Score (RAPS)}:
\begin{equation}\label{eq:raps}
\text{RAPS}(p, E) = \frac{p^{+}_{r}}{p^{+}_{r} + p^{-}_{r}}
\end{equation}
I.e., RAPS computes the percentage of relative approval votes.

\item
\textbf{$\beta$-Relative Absolute Majority Score ($\beta$-RAMS)}:
\begin{equation}\label{eq:rams}
\text{RAMS}(p, E) = 
    \begin{cases} 
    \text{RAPS}(p, E) & \text{if } p^+ \geq \beta \cdot N_a \\
    0 & \text{otherwise}
    \end{cases}
\end{equation}
That is, RAMS returns the RAPS value only if the proposal meets a minimum approval proportion (i.e., $\beta$), otherwise assigning zero (thus, it excludes paragraphs with insufficient participation).
 
\end{description}

% \subsection{Static Consensus-Conditioned Rule (static-CCR)}

% Next, we examine in depth the concept of defining a single static consensus condition, termed the static-CC rule, which uses the scores resulting from the static CSFs presented above.

% \begin{definition}[static-CCR]
%     Given an event list $E$, static-CSF $C_s$ and threshold value $x$ defining a Condition $(C_s(p, \mathrm{stance}(\mathcal{T})), x)$, the \emph{static-CCR} rule determines that for each paragraph $p \in P(E)$ if $C(p, \mathrm{stance}(\mathcal{T})) \geq x$ then $p \in \mathcal{R}(E)$.
% \end{definition}

% \begin{table}[t]
%     \centering
%     \caption{\label{tab:static_example_solutions} ($\mathcal{R}(E)$) for the static CCRs of Example~\ref{example:static ccr}.}
%     \begin{tabular}{c|c|c}
%         {Static-CCR}           & {$\mathcal{R}(E)$}         & $S_{\mathcal{R}}$   \\ \h
%         ($APS(p, E), 0.5$)  & $\{p_1, p_2, p_3, p_4\}$ & $\langle 1,1,1,1\rangle$    \\
%         ($RAPS(p, E), 0.5$)       & $\{p_1, p_2, p_4\}$ & $\langle 1,1,0,1 \rangle$ \\
%         ($0.25-RAMS(p, E), 0.5$)  & $\{p_1, p_4\}$ & $\langle 1,0,0,1\rangle$    \\
%         ($0.5-RAMS(p, E), 0.5$)   & $\{ \}$         & $\langle 0,0,0,0 \rangle$
%     \end{tabular}
% \end{table}

\begin{example}\label{example:static ccr}
Consider an instance $\mathcal{T}$ with the following event list:
\begin{align*}
  E = \langle & (a_1, p_1, +1), (a_2, p_1, +1), (a_3, p_1, +1), (a_1, p_2, +1) \\
              & (a_4, p_1, -1), (a_5, p_1, -1), (a_2, p_3, +1), (a_1, p_3, +1) \\
              & (a_3, p_3, -1), (a_3, p_4, +1), (a_1, p_4, +1), (a_5, p_4, -1) \\
              & (a_2, p_4, +1), (a_4, p_4, +1), (a_2, p_4, 0),  (a_4, p_3, -1) \rangle.
\end{align*}
The stance of this instance is described in Table~\ref{tab:stance_matrix}. Accordingly, the scores of each paragraph $p \in P$ -- corresponding to different static CSFs and the solutions -- corresponding to these CSFs with a threshold $x = \nicefrac{1}{2}$ -- are given in Table~\ref{tab:static_example}.
\end{example}

\begin{table}[t]
    \caption{\label{tab:stance_matrix} Stance matrix  of the instance presented in Example~\ref{example:static ccr}.}
    \centering
    \begin{tabular}{c@{\hspace{4mm}}cccc}
    \toprule
    \multirow{1}{*}{Agent / Paragraph} & \multirow{1}{*}{$p_1$} & \multirow{1}{*}{$p_2$} & \multirow{1}{*}{$p_3$} & \multirow{1}{*}{$p_4$} \\    
    \midrule
    $a_1$ & +1 & +1 & +1 & +1 \\
    $a_2$ & +1 & 0 & +1 & 0 \\
    $a_3$ & +1 & 0 & -1 & +1 \\
    $a_4$ & -1 & 0 & -1 & +1 \\
    $a_5$ & -1 & 0 & 0 & -1 \\
    \bottomrule
    \end{tabular}
\end{table}

\begin{table}[t]
    \caption{\label{tab:static_example} Static CSFs and CCRs of the instance of Example~\ref{example:static ccr}.} 
    \begin{tabular}{l@{\hspace{4mm}}cccc@{\hspace{4mm}}c}
    \toprule
     \multirow{1}{*}{Static-CSF} &  \multirow{1}{*}{$p_1$}  &  \multirow{1}{*}{$p_2$}   &  \multirow{1}{*}{$p_3$}    &  \multirow{1}{*}{$p_4$}   &  \multirow{1}{*}{$ [\text{Static-CSF}, 0.5]$}\\ 
    \midrule
    $\mathrm{APS} $          & $0.60$ & $1.00$  & $0.50$   & $0.75$  & $\{p_1, p_2, p_3, p_4\}$\\
    $\mathrm{RAPS}$         & $0.52$ & $1.00$  & $0.47$   & $0.65$  & $\{p_1, p_2, p_4\}$\\
    $0.25\mathrm{-RAMS}$    & $0.52$ & $0.00$  & $0.47$   & $0.65$  & $\{p_1, p_4\}$\\
    $0.75\mathrm{-RAMS}$     & $0.0$  & $0.00$ & $0.00$  & $0.0$    & $\{ \}$\\
    \bottomrule
    \end{tabular}
\end{table}

The static-CCR $\mathrm{CCR}_{[\mathrm{RAPS}, 0.5]}$ is an RM rule and thus maximizes social welfare.
As for stability, next we show that no static CCR is stable (we assume an additional notion -- \textit{versatility} -- to rule out naive cases, e.g., a static-CCR that excludes all paragraphs and is thus trivially stable; versatility requires $\mathcal{R}$ to be able to produce any output).

\begin{definition}[Versatility]
\label{definition:versatility}
An aggregation rule $\mathcal{R}$ satisfies \emph{Versatility} if, for each subset of paragraphs $P' \subseteq P$, there exists an instance $\mathcal{T}$ such that $\mathcal{R}(\mathcal{T}) = P'$.
\end{definition}

% This axiom ensures that the aggregation method $\mathcal{R}$ can produce any subset of paragraphs as the solution. 

\begin{theorem}[\appref{thm:NoStable}]\label{thm:NoStable}
  No versatile static CCR rule is stable.
\end{theorem}

\begin{proof}\label{proof_SI_Stable}
Towards a contradiction, let $\mathcal{R}$ be a versatile static-CCR rule that is stable.
Consider an initial event list $E_1 = \emptyset$. By stability, there exists an event list $E_1^{'}$ such that for any additional event list $E_1''$, $\mathcal{R}(E_1 + E_1^{'}) = \mathcal{R}(E_1 + E_1^{'} + E_1^{''}) = P_1$.
I.e., for the event list $E_1$ exists another event list $E_1'$ such that the aggregation method $\mathcal{R}$ will produce the same solution $P_1$ (a list of paragraphs) regardless of any additional events $E_1^{''}$ occurring in the future.
Let $P_2$ represent another possible solution, different from $P_1$ (i.e., $P_1 \neq P_2$). By versatility, there is an event list $E_2$ with $\mathcal{R}(E_2) = P_2$. Thus, the aggregation method $\mathcal{R}$ can produce a different solution $P_2$ for some event list $E_2$.
Define $E_1^{'''}$ to be the so-called ``Undo'' event list of $(E_1 + E_1^{'})$: this means that after adding the events in~$E_1^{'''}$ to the event list $(E_1 + E_1^{'})$, the extended event list has the same stance matrix, consisting of only abstention votes. Therefore, $\mathrm{stance}(E_1 + E_1^{'}+E_1^{'''})_{a, p} = 0$ for all $a$ and $p$.
Consider now the event list $E^{*} = E_1 + E_1^{'} + E_1^{'''} + E_2$, $\mathrm{stance}(E^{*}) = \mathrm{stance}(E_2)$. According to Stability, $\mathcal{R}(E^{*}) = P_1$. Since $\mathcal{R}$ is a static-CCR, and as $\mathrm{stance}(E^{*}) = \mathrm{stance}(E_2)$, we get that $\mathcal{R}(E^{*}) = \mathcal{R}(E_2)= P_2$. This, however, contradicts the fact that $P_1 \neq P_2$.
\end{proof}

\appendixproof{thm:NoStable}{
\begin{figure}[ht]
    \centering
    \includegraphics[width=1\columnwidth]{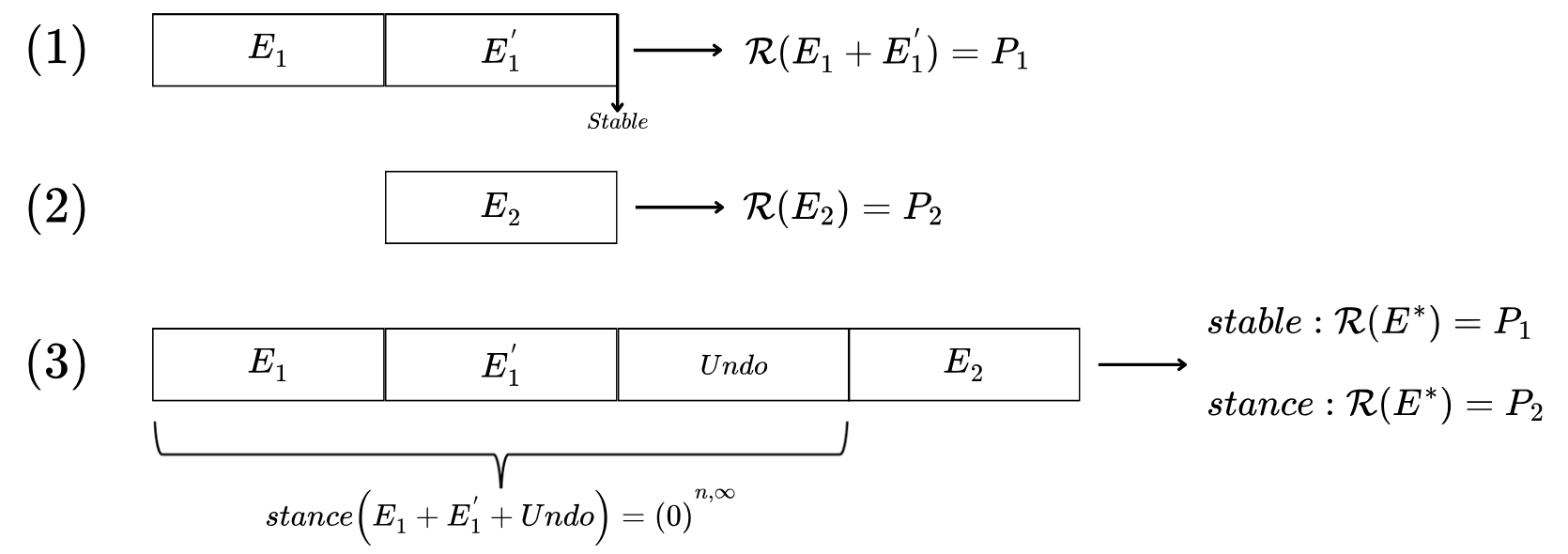} 
    %\caption{Schematic representation of the proof of Theorem~\ref{thm:NoStable}.}
    \caption{Schematic representation of the proof of \texorpdfstring{Theorem~\ref{thm:NoStable}}{Theorem}.}
    \label{fig:flow_of_NoStable}
    \nicespace
\end{figure}
}

% \begin{corollary}
%     No static-CCR can be both stable and SWM.
%     \nimrod{Yeah, give a proof for the corollary -- in it, explain why it follows from the stuff above}
% \end{corollary}

% \paragraph{Trade-Off}

% \begin{theorem}\label{thm:mutually_exclusive}
%     No static-CR can be both stable and SWM.
% \end{theorem}

\subsection{Dynamic CCRs}

Following Theorem~\ref{thm:NoStable}, and as stability is crucial for our setting, we go on to consider non-static CCRs.
Correspondingly, we consider CCRs corresponding to non-static CSFs (recall Definition~\ref{staticCSF}) -- in particular, these rules use not only the stance matrix to assign scores to paragraphs, 
%such that assign scores to paragraphs not only based on the stance matrix, 
but also the event list itself. 
As a result, the aggregation results are influenced by the specific sequence and composition of events rather than just their cumulative summary. 
%
% \begin{definition}[Dynamic-CSF]
%     \label{def:dynamicCSF}
%     Let $\mathcal{T} =(A, E)$ be a document writing instance. A \textbf{dynamic-CSF}, denoted as $C_d$, is a CSF function such that
%     $\forall p \in P \text{, } C_d : (p, E) \to [0, 1]$.
% \end{definition}
%
% In contrast to static-CSF, dynamic functions calculate the score for each candidate paragraph dynamically, considering the intricate details of the event list itself. This dynamic scoring approach brings about a crucial distinction – the aggregation results are influenced by the specific sequence and composition of events rather than just their cumulative summary. 
%
% \begin{theorem}
%     \label{thm:static_contained_in_dynamic}
%     The class of static-CCR aggregation rules is strictly contained in the class of dynamic-CCR rules.
% \end{theorem}
%
% \begin{proof}
%     Let $\mathcal{R}_d$ be a \emph{dynamic-CCR} rule defined by a dynamic-CSF noted as $C_d$ and threshold $x$. We can define a new static-CSF noted as $C_s$, such that $C_s(p, \mathrm{stance}(\mathcal{T})) = C_d(p, E)$ because given any event list $E$ the summary $\mathrm{stance}(\mathcal{T})$ can be calculated. As a result, we can represent the matching \emph{static-CCR} rule as $\mathcal{R}_s := (C_s,x)$ for any \emph{dynamic-CCR} rule.
% \end{proof}
%
(Note that, indeed, the class of static-CCRs is strictly contained in the class of dynamic-CCRs; e.g., a non-static CCR that is versatile and stable -- contrasting Theorem~\ref{thm:NoStable}, is given in Appendix~\ref{appendix:versatile_stable}.)

\subsection{Dynamic Properties}

% The example above implies that, indeed, not all CCRs are static -- and demonstrates a stable CCR.
%
Next, we consider a rich family of CCRs that are non-static (and thus, may be stable) to establish a structured design space of rules to evaluate.
The intuitive idea is to consider static CSFs together with functions that transform these static CSFs into dynamic ones by incorporating a control parameter $param$. This parameter relates to dynamic properties of the event list~$E$ (or the solution $\mathcal{R}(E)$) -- e.g., the number of events ($|E|$) -- and is used to account for the progression of this $param$’s value when assigning scores to paragraphs.
% and is being used to account the progression of the value of $param$ when assigning scores to paragraphs. 
As we will see, this will enable us to consider rules that, e.g., as time progresses (i.e., as events are added to the event list), require a higher and higher vote majority for a paragraph to be included in the solution.
Some examples of dynamic parameters are:
\begin{itemize}

\item
\textit{Form parameters}:
number of events $|E|$. %(corresponding to the interaction level in the system).

\item
\textit{Content parameters}:
number of agents involved $|A(E)|$; number of proposed paragraphs $|P(E)|$.

\item
\textit{Document parameters}:
number of paragraphs in the solution $|\mathcal{R}(E)|$; number of updates $|U(\mathcal{R}, E)|$ (defined next).

\begin{definition}[Number of updates]\label{definition:updates}
  Given an event list $E = \langle e_1, e_2, \ldots, e_{|E|} \rangle$ and an aggregation rule $\mathcal{R}$, let $\mathcal{R}^*$ be the sequence of solutions produced by $\mathcal{R}$ for all event prefixes (i.e., $\mathcal{R}^* = \langle \mathcal{R}(\langle e_1 \rangle), \mathcal{R}(\langle e_1, e_2 \rangle), \ldots, \mathcal{R}(E) \rangle$); then, the \textbf{number of updates} of rule $\mathcal{R}$ for event list $E$ is defined as:
  \[
    U(\mathcal{R}, E) = \left| \left\{ i \in [|E| - 1] : \mathcal{R^*}(i + 1) \neq \mathcal{R^*}(i) \right\} \right|.
  \]
\end{definition}

\end{itemize}

Next, we discuss how to incorporate the values of these parameters in the design of the aggregation rules: we distinguish between \textit{harsh rules} and \textit{smooth rules}.

% Consequently, we can distinguish between inclusion and omission updates and accordingly create count properties - the \emph{number of inclusions} and \emph{number of omissions}.
% \end{itemize}

\subsection{Harsh Rules}\label{section:dynamizers}

% \begin{definition}[Dynamizer]
% \label{Dynamizer}
%     Let $C_s$ be a static-CSF and $t$ be a natural parameter. A \textbf{Dynamizer} is a dynamic-CSF parameterized by $t$ and noted as $t$-$C_s$. This parameter $t$ defines a dynamic property in the static function that governs the subset of events the function considers. Formally,
%     \[
%         t\text{-}C_s : (p, E, t) \to [0, 1]
%     \]
% \end{definition}\nimrod{I've found this definition very hard to read somehow..}

% An important note is that when $t \rightarrow \infty$, $t\text{-}C_s$ becomes identical to $C_s$, thus the CSF returns to be static \nimrod{I don't see this implied by the definition}. 

We begin with the class of so-called \emph{harsh rules} (that correspond to \emph{harsh-CSFs}, a class of non-static CSFs) -- intuitively, these rules completely (harshly) stabilize as soon as the dynamic parameter at hand reaches a predefined value.
Definition~\ref{definition:harsh_param} formalizes this notion.
% Consider the following definition.

% Concerning control over dynamic properties, We want to differentiate between two type: \textbf{\emph{Harsh and Smooth}} Dynamizers. A \emph{Harsh} dynamizer involves toughening the scoring mechanism at a specific point determined by the community. Conversely, a \emph{Smooth} dynamizer entails a more gradual and subtle toughening of the scoring system.

\begin{definition}
\label{definition:harsh_param}
Let $E = \langle e_1, \ldots, e_n\rangle$ be an event list, $param$ a dynamic parameter, and $t$ a natural number. Let $j$ be the smallest index such that $param(\langle e_1, \ldots, e_j \rangle) = t$, and $\infty$ if no such index exists.
Then, let $E\!\downarrow_{param, t}$ be $\langle e_1, \ldots, e_j \rangle$.
\end{definition}

That is, $E\!\downarrow_{param, t}$ is the prefix of $E$ corresponding to the first event for which the value of $param$ reaches $t$. With this, we are ready to define harsh CSFs (HCSFs).

\begin{definition}[Harsh CSF]
\label{definition:HarshD}
    Let $C_s$ be a static-CSF, $param$ a dynamic parameter, and $t$ a natural number. Then, the \emph{harsh-CSF} $\mathrm{HCSF}_{[\mathrm{C_s}, param, t]}$ is a dynamic-CSF with $\mathrm{HCSF}_{[\mathrm{C_s}, param, t]}(p, E) = \mathrm{C_s}(p, E\!\downarrow_{param, t})$.
\end{definition}

\begin{example}
The harsh-CSF $\mathrm{HCSF}_{[\mathrm{APS}, |E|, 7]}$ computes, for each paragraph, its upvote
%up vote
fraction, but considering only the first $7$ events. Similarly, the harsh-CSF $\mathrm{HCSF}_{[\mathrm{APS}, |U|, 7]}$ computes the 
%up vote
upvote fraction, but only up until the $7$th output change event.
\end{example}

% - the dynamic-CS function takes into account only the first $t$ events. Formally, for all paragraph $p\in P$ and any event list $E$:
%     \[
%         HarshD(C_s, |E|_t) := C_s(p, \mathrm{stance}(E[t_{events}]))
%     \]
%     where $|E|_t$ is the number of events index and $E[t_{events}] = \langle e_1, ...,e_t \rangle$ the sub-list of the first $t$ events.
    
%     \item $HarshD(C_s, |U|_t)$ - the dynamic-CS function considers the first $t$ solution list updates (see Definition~\ref{definition:updates}). Formally, for all paragraph $p\in P$ and any event list $E$:
%     \[
%         HarshD(C_s, |U|_t) := C_s(p, C_s(p, \mathrm{stance}(E[t_{updates}])))
%     \]
%     where $|U|_t$ is the number of updates index and $E[t_{updates}] = \langle e_1, ...,e_i\rangle$ the sub-list of $E$ containing all events up to and including the $t$-th update occurring in event $i$. 

%     \item $HarshD(C_s, |P(E)|_t)$ - the dynamic-CS function dynamically acknowledges all events until there are $t$ paragraphs suggested. Formally, for all paragraph $p\in P$ and any event list $E$:
%     \[
%         HarshD(C_s, |P(E)|_t) := C_s(p, \mathrm{stance}(E[t_{suggestions}])) 
%     \]

%     where $|P(E)|_t$ is the number of suggestions index and $E[t_{suggestions}] = \langle e_1, ...,e_i\rangle$ the sub-list of $E$ containing all events up to and including the $t$-th update occurring in event $i$. 
% \end{description}

Indeed, some CCRs corresponding to harsh CSFs are stable (consider, e.g., $\mathrm{HCSF}_{[\mathrm{APS}, |E|, 7]}$ -- that completely stabilizes after $7$ events).
%; this, however, is also their disadvantage -- their use of the dynamic parameter $param$ is too strict (consider also work on constitution stability~\cite{abramowitz2021amend,FoundingAmendingConstitution}). This also makes them susceptible to adversarial actions: e.g., 
However, their stability-enforcing rigidity is also a disadvantage -- the strict use of the dynamic parameter $param$ makes them susceptible to adversarial actions (consider also work on constitution stability~\cite{abramowitz2021amend, FoundingAmendingConstitution}).
For instance, an agent sending $7$ events for a CCR that corresponds to $\mathrm{HCSF}_{[\mathrm{APS}, |E|, 7]}$ prevents all other agents from affecting the solution. (From a practical point of view, though, Harsh rules are important when we absolutely need to take an action at a certain point in time, e.g., for certain budgeting decisions.)

\subsection{Smooth Rules}\label{section:Smooth_def}
Smooth rules utilize the $param$ value and are less harsh than harsh rules. This is achieved by a \emph{smoothing function}, parameterized by a \emph{smoothing parameter} $\alpha\in [0, 1]$ that modulates the impact of the $param$ value on the consensus score: a higher $\alpha$ means a harsher rule, 
%and
while a lower~$\alpha$ means a smoother dependence on the $param$ value. The goal is to explicitly adjust the \emph{trade-off} between social welfare and stability. 
We begin with the following definition.

\begin{definition}[Dynamic Smoothing Function]
\label{SmoothD_Function}
    Let $F_\alpha: [0,1] \times \mathbb{N} \to [0,1]$ be a family of functions parameterized by $\alpha \in [0,1]$. A \textbf{Dynamic Smoothing Function}, denoted as $F_\alpha(x, t)$, is defined for each $\alpha \in [0,1]$, where $x \in [0,1]$ and $t\in \mathbb{N}$. 
\end{definition}

We are ready to define Smooth CSFs (SCSFs).

\begin{definition}[Smooth CSF]
\label{definition:smoothD}
    Let $F_{\alpha}$ be a smoothing function,
    $C_s$ be a static-CSF, $param$ a dynamic parameter, and $t$ a natural number.
    Then, $\mathrm{SCSF}_{[\mathrm{C_s}, param, t, F_{\alpha}]}$ is a dynamic-CSF. Formally, for all paragraph $p\in P$ and any event list $E$: $\mathrm{SCSF}_{[C_s, param, t, F_{\alpha}]}:= F_{\alpha}(C_s(p, E), t_{param}) \to [0, 1]$.
\end{definition}

The smoothing function modulates the consensus values, depending on the $param$ value $t$. Intuitively, we are looking for smoothing functions that are: (1) monotonically decreasing in $t$; (2) monotonically decreasing in $\alpha$; (3)  monotonically increasing in~$x$; and (4) returning $0$ only if the initial score is $0$ (i.e., not completely stabilizing). 
Next, we proceed with one type -- exponential smoothing functions -- as they effectively model a gradual increase with asymptotic behavior as the system evolves through $param$.
In this full version, we formally define these properties and prove that exponential smoothing functions, described next, satisfy them (see Appendix~\ref{appendix:appendix_Smoothing_functions}).

\begin{description}
    
    \item
    \textbf{Exponential Smoothing Functions}:
    \[
    F^{exp}_{\alpha}(x, t) = x \cdot e^{-t \cdot \alpha \cdot (1 - x)}
    \]
\end{description}

%
% Intuitively, we are looking for a smoothing function that satisfy  to qualify as a consensus scoring function (CSF), ensuring both the integrity of its values and a coherent relationship with the defining parameters. The following characteristics out these necessary conditions:
%

% Recognizing that different smoothing functions produce varying relationships between the scores assigned by the consensus scoring function (CSF) and the dynamic features based on the mathematical structure of the function, it becomes essential to explore multiple functions. This examination allows for a deeper understanding of how the parameters influence the scoring mechanism and its dynamic behavior. 

\begin{example}
Let $param$ be the number of events ($|E|$); use an $exp$ smooth function with $\alpha = 0.1$ ($F^{\exp}_{\alpha = 0.1}$); and consider the static-CSF $\mathrm{APS}$. The corresponding dynamic-CSF is: $\mathrm{SCSF}_{[\mathrm{APS}, |E|, t, F^{\exp}_{\alpha = 0.1}]}(p, E) = \mathrm{APS}(p, E) \cdot e^{-t_{|E|} \cdot 0.1 \cdot (1 - \mathrm{APS}(p, E))}$.
E.g., using $p_1$ from Example~\ref{example:static ccr}, with $\mathrm{APS}(p_1, E) = 0.60$ (see Table~\ref{tab:static_example}) and $t_{|E|} = 16$: 
$\mathrm{SCSF}_{[\mathrm{APS}, |E|, t, F^{\exp}_{\alpha = 0.1}]}(p_1, E) = 0.60 \cdot e^{-16 \cdot 0.1 \cdot (1 - 0.60)} = 0.316$.
Hence, the stance of $p_1$ -- (3+, 2-) -- no longer meets the inclusion criteria of $0.5$ (see Table~\ref{tab:stance_matrix}).
\end{example}

% We go on to do simulation-based analysis of these smooth rules.
% ; note that, indeed analytically it is too hard, as there are too many moving parts.

% As HarshD scoring functions are not satisfactory (due to their hachability), we wish to relax the stability constraint and consider more fine-tuned methods. In particular, we observe the inherent \emph{trade-off} between fully dynamic and stability, and we wish to find rules that strike good balance.

\section{Experimental Analysis}\label{section:simulations}

We report on simulations examining the satisfaction–stability tradeoff across rules and settings.

% The theoretical analysis above led us to consider the specific aggregation rules defined above. Next, we analyze them using simulation.
%
% We discuss the experimental design (Section~\ref{section:experimental_design}) and then report on the results (Section~\ref{section:experimental_results}).

\subsection{Experimental Design}\label{section:experimental_design}

We describe the experimental setup and the input settings used in our simulations 
(see Figure~\ref{fig:simulation_overview} in Appendix~\ref{appendix:experiments} for an illustration):
\begin{itemize}

\item
\textbf{Context} -- either a fixed set of possible paragraphs out of which the simulated collaborative constitution is to be written; or, in more advanced simulations, a certain thematic topic.

\item
\textbf{Population} -- we utilize several agent modeling approaches to simulate different populations interacting to create a collaborative text document: unstructured (Section~\ref{section:unstructured_population}), Euclidean (Section~\ref{section:euclidean_population}), and LLM-based (Section~\ref{section:textual_population}). 

\item
\textbf{Scheduler} -- this entity \emph{runs the show}: it interacts with the agents, essentially deciding on the order by which different population members interact with the system (a useful and standard way of dealing with decentralized systems like ours).

\end{itemize}

Later on, it will be useful to describe the specific contexts, populations, and schedulers that we use in our simulations. Next, given such an event list, we continue the abstract and general description: 

\paragraph{Aggregation Rules.} As we constructed three types of CCRs (static, harsh, and smooth), we observe the following rules with threshold $x =\nicefrac{1}{2}$.\footnote{We chose $threshold = \nicefrac{1}{2}$ as, in its simplest form, it corresponds to majority voting; note also that, while other values of $threshold$ makes sense as well, different CSFs can simulate equivalent rules (thus, there is no loss of generality here, but mainly convenience).}
Three static-CSFs are evaluated: $\mathrm{CSF} \in$ \{APS, RAPS, $\beta$-RAMS\}, with $\beta \in \{0.05, 0.1, 0.3, 0.5\}$. The dynamic property types are $param \in $ \{events ($|E|$), paragraphs ($|P(E)|$)\}. For Harsh-CCRs, temporal cutoffs are set as $t \in \{0, 50, 100, 150\}$. For Smooth-CCRs, we apply exponential smoothing functions $F^{\text{exp}}_{\alpha}$ with $\alpha \in \{0.1, 0.3, 0.5, 1.0\}$. In total, initially 54 distinct CCR rule configurations are evaluated, and later on, 64 for the LLM-based simulation model. (Note that, for space considerations, we do not consider all parameter combinations.)

\mypara{Evaluation Metrics}
We focus on evaluation metrics that relate to the \textbf{stability} of the process and the \textbf{satisfaction} with the output. We use the following metrics, normalized to $[0, 1]$:
%
% \paragraph{Stabilization Metrics.} Measure the extent to which the solution variance decreases as the number of events increases, indicating convergence toward a final document. We examined the following stabilization metrics:
%
\begin{itemize}
    \item \emph{Average Number of Document Updates:} The average number of times the document is updated (see Definition~\ref{definition:updates}):
$stability_{[updates]}(\mathcal{T}, S) = \frac{|E| - U(\mathcal{R}, E)}{|E|}$.
    \item \emph{Normalized social welfare:} The community's satisfaction $Sat_A(\mathcal{T}, S)$ (see Definition~\ref{Solution_Satisfaction}) is divided by the number of active agents $ N_{active} := \sum_{a=1}^{n} \mathbb{1}_{\{\exists p : \mathrm{stance}(\mathcal{T})_{a,p} \neq 0\}}$.
    % hence:
      % \[
    % Sat_A(\mathcal{T}, S) = \frac{\sum\limits_{a=1}^{n} Sat_a(\mathcal{T}, S)}{\sum_{a=1}^{n} \mathbb{1}_{\{\exists p : \mathrm{stance}(\mathcal{T})_{a,p} \neq 0\}}}
  % \]
\end{itemize}

% We have considered various populations, discussed below: unstructured, Euclidean, and LLM-based.
% %
% \begin{itemize}

% \item
% \emph{Unstructured}: Here, there are no structural restrictions or imposition on the agents; essentially, agent act uniformly at random. While surely not realistic, this very simplistic setting highlights some of the basic properties of our setting.

% \item
% \emph{Euclidean}: Here, paragraphs are artificially embedded in a Euclidean space; this corresponds to, say, writing a text document on some one-dimensional topic such as tax.

% \item
% \emph{Textual}: Here, agents are modeled using LLMs, and reason about texts, following certain prompts.

% \end{itemize}

\subsection{Unstructured Population}\label{section:unstructured_population}

The setting here is as follows:
\begin{itemize}

\item
\textit{Context}: we create a set of paragraphs $P$, referred to by their index~$p_i$. (Here, the agent population is agnostic to text.)

\item
\textit{Population}: each agent, when approached by the scheduler, can trigger one of three events with equal probability. The first type of event is a new suggestion case where the agent suggests a new paragraph (with index $max + 1$). The second is a preference vote event, in which the agent randomly uniformly at random selects an existing paragraph and casts a preference vote, either up-voting or down-voting. The third event allows the agent to revisit a previously preference-voted paragraph and cast an avoidance vote, changing its preference to neutral (0).

%Population: we defined agents where each agent, when approached by the scheduler, operates as follows: with probability $p_1$, it suggests a new paragraph (with index $max + 1$; with the complement probability, it select an existing paragraph $p$ uniformly at random, and then, with probability $p_2$ it likes it and with probability $1 - p_2$ it dislikes it.%

\item
\textit{Scheduler}: in each iteration, it selects an agent uniformly at random to generate a vote. After receiving a generated event, the scheduler adds the event to the list until it gets $t$ events. 

\end{itemize}

\begin{figure}[t]
    \centering
    \includegraphics[width=0.45\textwidth, height=0.3\textheight, keepaspectratio]
    {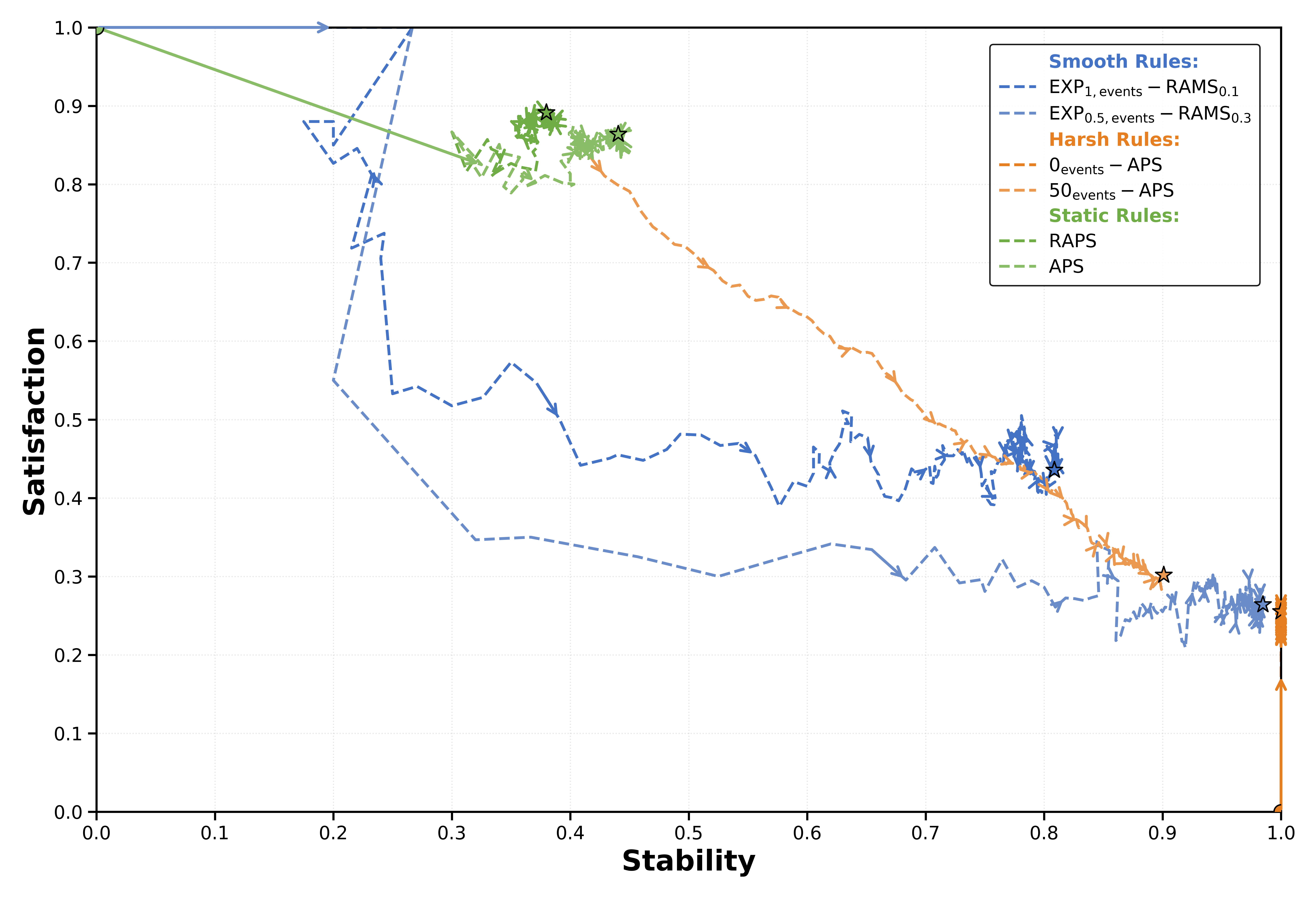}
    \caption{Stability ($x$-axis) and satisfaction ($y$-axis) trade-off evolution in  for different $\mathrm{CCR}$ rules family; instance with 20 unstructured agents for 300 events. For each rule, the plot shows an arrow that advances with the events and illustrates the stability/satisfaction values.}
    \label{fig:unstructured}
    \nicespace
\end{figure}

This agent model elucidates the evolution of the event list within the design trade-off space. As illustrated in Figure~\ref{fig:unstructured}, static rules promote system dynamism, yielding high satisfaction but low stability. Conversely, harsh rules enhance stability at the expense of satisfaction, while smooth rules exhibit an exponential increase in stability accompanied by a decrease in satisfaction.

% A core interest following the impossibility is to analyze the tradeoff between satisfaction and stability. Given wide document applications (fitting our settings) and rule configurations possibilities, analyzing the Pareto curve of dominated rules (I, achieved the highest score in at least one measure) shows nicely, given a certain system configuration, how the trade-off behaves.

\subsection{Euclidean Population}\label{section:euclidean_population}

The population model described above is primitive in that all agent decisions are taken uniformly at random. Here, we consider more involved agents, allowing us to control the level of consensus in the population. 
%
%\nimrod{perhaps it's nice to say something like: the population above is very primitive; here we are considering a bit more involved population -- still rather simple, but at least we are able to play with populations of different styles: one in which all agents are quite in sync, and one in which they are not quite but around some mean (the gaussian); and one in which there is an ideological divide..}
%
The setting is as follows:
\begin{itemize}

\item
\textit{Context}: each paragraph $p \in P$ corresponds to a point $0 \leq p \leq 1$; the intended meaning is that the position on the corresponds to a thematic sentiment score (e.g., $0$ means ``do not agree''; while $1$ means ``completely agree'').

\item
\textit{Population}: each agent $a \in A$ is defined by an interval $[left(a), right(a)]$, $0 \leq left(a) \leq right(b) \leq 1$, representing their sentiment range regarding the . 
% We explore Uniform, Gaussian, and Bimodal Gaussian Distributions to define the ideal point and radius.

\item
\textit{Scheduler}: in each event iteration, the scheduler randomly selects an agent $a$ from the population $A$. The agent action (suggestion/vote) depends on the ratio $r_a$, the fraction of previously proposed paragraphs within the agent’s score interval relative to the total number of suggestions. We set a minimum rate value $r_{min} = 0.2$ so if $r_a < r_{min}$, agent $a$ initiates a \textit{new proposal} $p$ within $[left(a), right(a)]$. Otherwise, the agent votes on a randomly selected paragraph $p$. If $\mathrm{stance}((A, E))_{a,p} \in \{+1,-1\}$, the agent obtains ($0$). Else, the agent up-votes if $p \in [left(a), right(a)]$; otherwise, they down-vote. 
% We schedule the event lists with $r_{min} = 0.2$.

\end{itemize}

We consider three distributions of agents' ideal point (the middle of their interval): Uniform, Gaussian, and Bimodal Gaussian.
All experiments include $20$ agents and $300$ events, with metric values averaged over $5$ repetitions. This setup balances computational feasibility, yet is sufficient to demonstrate the effects we are interested in (see Appendix~\ref{appendix:extended_pareto_analysis} for additional analyses).

These richer experiments allow for a \textbf{Pareto analysis} to rigorously assess trade-offs between satisfaction and stability -- two competing performance criteria central to our framework. Pareto-optimal rules are identified as those configurations for which no other rule provides simultaneously superior satisfaction and superior stability. This more succinct visualization (as we show only the Pareto non-dominated rules), as can be seen in Figure~\ref{fig:gaussian_comparasion}, enables an informed selection of aggregation strategies tailored to specific application needs (e.g., prioritizing stability for constitutional documents versus satisfaction for statements ones).
As the figure shows, when comparing the normal population (homogeneous) with the heterogeneous bimodal population, the latter exhibits a wider Pareto front than the homogeneous population, reaching higher satisfaction levels -- albeit at the cost of reduced stability. 
By contrast, the unimodal population yields a narrower, more tightly clustered Pareto frontier, indicating high stability yet moderate satisfaction.
% Also, the unimodal population yields a narrower frontier, tightly clustered around high stability and moderate satisfaction, 
This reflects consistent agreement among agents with overlapping preferences. 

\begin{figure}[t]
    \centering
    \includegraphics[width=0.45\textwidth, height=0.5\textheight, keepaspectratio]
    {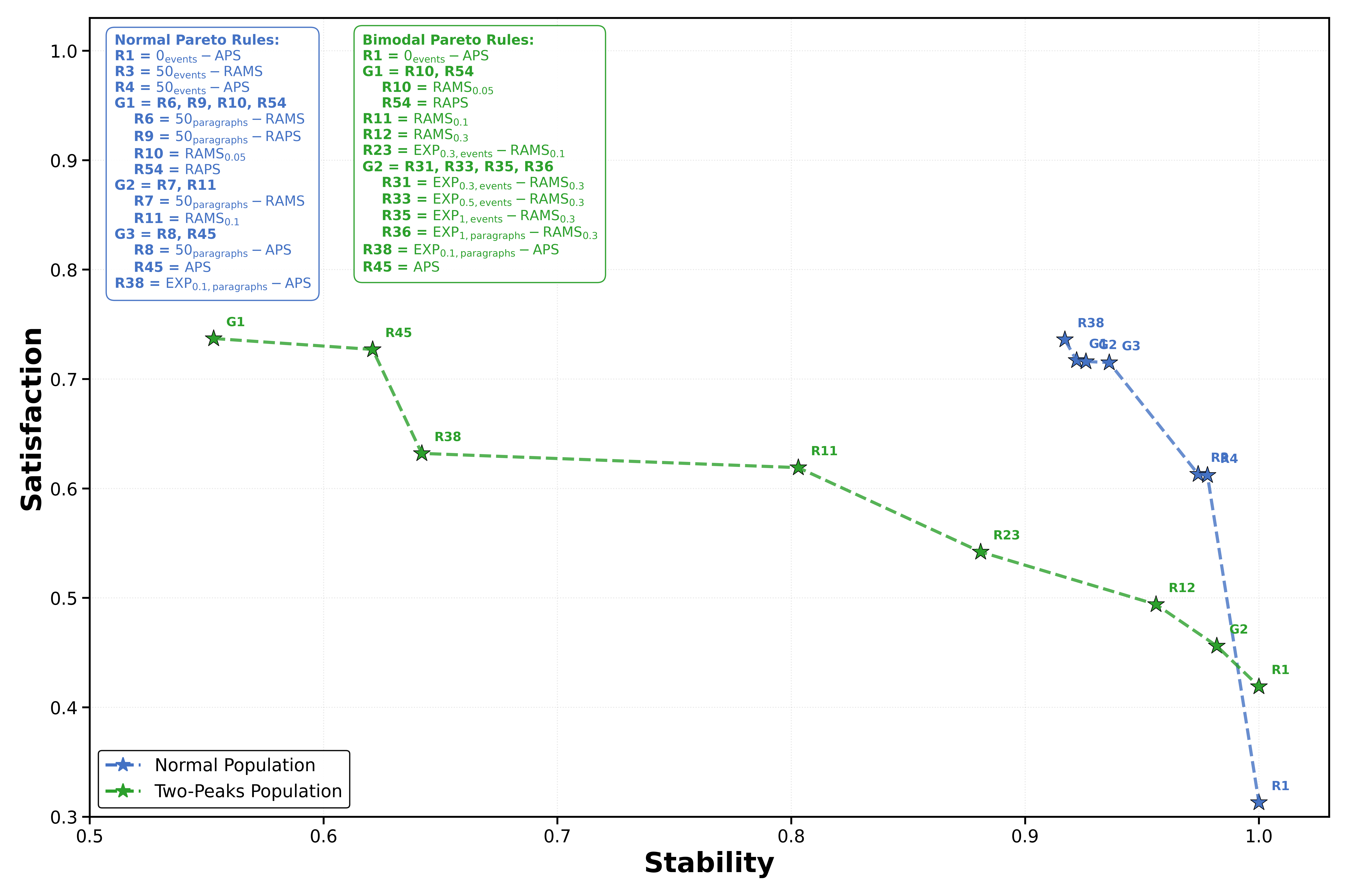}
    \caption{Comparison of the satisfaction-stability Pareto trade-offs achieved under two Euclidean populations: normal (unimodal) and two-peaks (bimodal); instances with $20$ agents and $300$ events.}
    \label{fig:gaussian_comparasion}
    \nicespace
\end{figure}

\subsection{LLM-Based Population}\label{section:textual_population}

To introduce greater behavioral realism, we simulate agents using large language models (LLMs). Utilizing prompt engineering techniques, these agents autonomously reason and respond contextually to textual content and evolving system states (stance and solution), significantly enhancing simulation depth and realism~\cite{userbehaviorsimulationlarge, simagentlargelanguagemodels}.

\paragraph{Population.} Each agent is defined by a demographic identity and sentiment orientation, and interacts with an evolving shared document through text-based actions. \textbf{Agent demographic} profiles, reflecting a real-world population data structure of age, sex, and educational attainment drawn from national statistics for Israel in 2022 (see Appendix~\ref{appendix:experiments}).
Each profile is associated with a sampled topic-specific sentiment category, representing the agent’s level of activism on the policy topic -- from active resistance to proactive support -- and shaping the voting and proposal behavior.
The categories are treated as strategic behavior types, consistent with population-level modeling of sentiment~\cite{advancingSentiment}.

\paragraph{Scheduler.} Each simulation run is initialized with a predefined configuration: a community of agents, an aggregation rule, and a target number of events. Agents interact with the system asynchronously. At each iteration, one agent is selected at random to observe the current system state (tally matrix, self-stance, and current document) and generate an action, either a proposal or a vote. The process proceeds for a fixed number of event iterations or until a convergence criterion is met. Note that due to LLM context-length limitations, we design the model so that an agent’s decision is based solely on the stance, rather than on the full history of events. This implementation assumes a uniform engagement rate for agents (see Section~\ref{section:outlook} for limitations).

\paragraph{Prompting and Action Generation.} 
Agent actions are produced using the GPT-4o-mini model, executed via LangChain's prompt interface (at a fixed temperature of 0.7). Each prompt consists of two components: a \underline{system prompt} encoding the agent's demographic profile and sentiment orientation, and a dynamically generated \underline{decision prompt} to guide the next interaction. The decision prompt uses a stepwise reasoning format adapted from Chain-of-Thought prompting (CoT) for reasoning~\cite{promptEngineering}.
To align agent actions with their sentiment orientation, we use dynamic \emph{few-shot prompting}: at each decision point, a subset of three policy statements is sampled from a pre-generated corpus (filtered by the agent’s sentiment category and domain).
(see Appendix~\ref{appendix:prompts} for templates.)

\begin{example}\label{example:llm_schedule}
Consider an instance with $20$ LLM-based agents and $50$ events, simulated under $\mathrm{CCR}_{[\mathrm{APS}, 0.7]}$ (see additional details and event generation in Appendix~\ref{appendix:simulations_schedulers}). 
The textual use case is a simulation of a collaborative writing process for the development of a city-level constitutional document that addresses climate-related measures. 

The resulting tally matrix is presented in Table~\ref{tab:example_llm_tally}, and the resulting policy document includes only paragraphs $p_1$ and $p_3$ (paragraph $p_2$ was excluded, as $\mathrm{APS}(p_2, E) = \nicefrac{5}{9} < 0.7$).

\end{example}

\begin{table}[t]
    \centering
    \footnotesize
    \caption{Tally matrix for $20$ LLM-based agents (see Example~\ref{example:llm_schedule}).}
    \setlength{\tabcolsep}{3pt} % reduce column padding
    \begin{tabular}{c@{\hspace{1mm}}p{4.6cm}@{\hspace{1mm}}c@{\hspace{1mm}}c}
    % \begin{tabular}{c@{\hspace{4mm}}cccc}
    \toprule
    \multicolumn{1}{c}{\textbf{Paragraph}} & \multicolumn{1}{c}{\textbf{Text}} & \multicolumn{1}{c}{\textbf{+ Votes}} & \multicolumn{1}{c}{\textbf{- Votes}} \\
    \midrule
    \multirow{2}{*}{$p_1$} & Implement community gardens in urban areas to promote local food production and biodiversity. & \multirow{2}{*}{$12$} & \multirow{2}{*}{$5$} \\
    \midrule
    \multirow{2}{*}{$p_2$} & Maintain current agricultural practices without promoting new sustainability initiatives. & \multirow{2}{*}{$5$} & \multirow{2}{*}{$4$} \\
    \midrule
     \multirow{2}{*}{$p_3$} & Create community workshops on sustainable gardening practices and local food systems. & \multirow{2}{*}{$3$} & \multirow{2}{*}{$0$} \\
    \bottomrule
    \end{tabular}
    \label{tab:example_llm_tally}
\end{table}

We run simulations with these configurations:
\begin{itemize}
    \item \textbf{Simulations:} The population consists of 20 agents. Each run includes 150 or 250 scheduled events, with five repetitions per aggregation rule configuration.
   \item \textbf{Datasets} \normalfont{(see Appendix~\ref{appendix:datasets} for details)}: Agent profiles are sampled from Israeli national census data (2022), capturing distributions over sex, age, and educational attainment. To support sentiment conditioning, we generate a synthetic corpus of climate policy proposals using a language model (GPT-4o), grouped into five sentiment categories. These serve as few-shot examples to align agent behavior with their orientation.
\end{itemize}

\begin{figure}[t]
    \centering
    \includegraphics[width=0.45\textwidth, height=0.5\textheight, keepaspectratio]{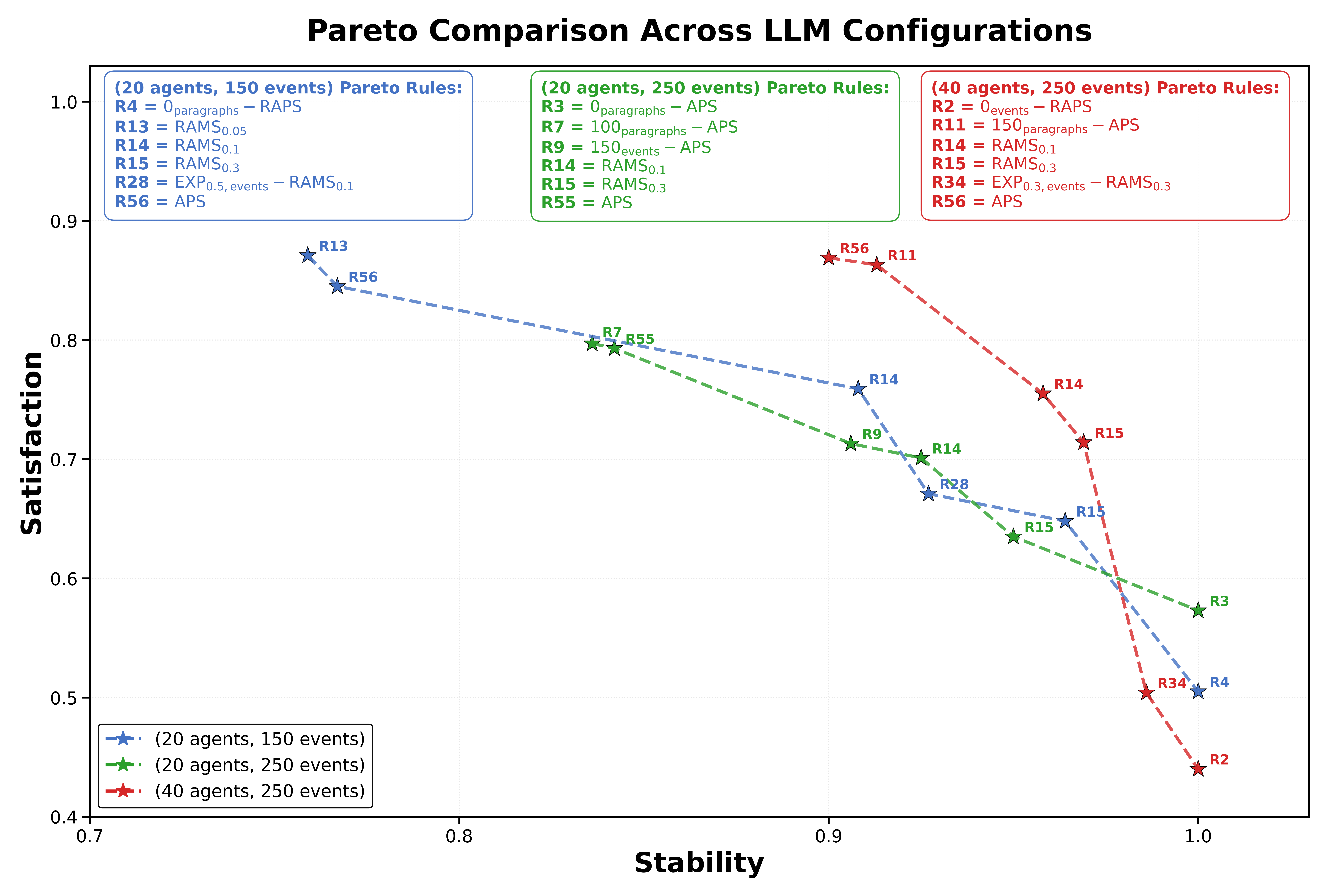}
    \caption{Comparison of the satisfaction–stability Pareto trade-offs achieved under three LLM-based simulation configurations: (20 agents, 150 events), (20 agents, 250 events), and (40 agents, 250 events).}
    \label{fig:configuration_comparasion}
    \nicespace
\end{figure}

Our Pareto analysis (Figure~\ref{fig:configuration_comparasion}) highlights crucial dependencies between simulation parameters (event length and population size) and the resulting stability–satisfaction trade-offs. Shorter event lists inherently favor stability, while extended simulations offer broader satisfaction-oriented configurations. Similarly, increasing community size enhances satisfaction due to richer agent diversity but challenges stability. APS and RAMS exhibit diversity of Pareto rules across scenarios, providing versatile tools for various contexts.

\section{Outlook}\label{section:outlook}

We described a model for collaborative text aggregation and showed how different aggregation rules shape system dynamics, highlighting the stability–satisfaction tradeoff.
\textbf{Methodologically}, our simulation framework has proven to be effective (e.g., showing that uniformly distributed populations result in notably larger solutions; see Appendix~\ref{appendix:additional_results}). 
\textbf{Practically}, our work enables a community to choose an aggregation method based on the point in the satisfaction/stability Pareto curve they prefer. 
Directions for future research include:
\begin{itemize}
    \item Exploring methods (e.g., recommendation systems) to deal with agents voting only on a few paragraphs -- partial preferences; 
    \item Considering \textbf{open community} settings in which agents gradually enter and leave the system (e.g., with preferential attachment);
    \item Designing \textbf{further aggregation rules}, including those that consider multiple paragraphs simultaneously, to mitigate semantic inconsistencies; 
    % and processes 
    additionally, exploring processes that incorporate token-based incentives and interaction gamification;
    \item Performing more \textbf{simulations} as well as \textbf{user studies};
    \item Performing theoretical analysis on smooth rules;
    \item Adapting our framework \textbf{other use-cases} of collaborative document writing (e.g., ordered sets).
\end{itemize}

\newpage
\clearpage

%% Use this environment to include acknowledgements(optional).

\begin{ack}
Co-funded by the European Union under the Horizon Europe project PERYCLES (Grant Agreement No. 101094862). The authors thank Aharon Porath for numerous discussions on the Consenz platform.
\end{ack}

\bibliography{bib-ecai25}

\newpage
\clearpage
\appendix

% We provide missing proofs, more examples and figures,m details regarding the smoothing function selection, and experimental details.

\section*{Appendix Overview:}

\begin{enumerate}[label=\Alph* \textendash, leftmargin=2em, labelsep=0.5em]
    \item \nameref{appendix:missing_proofs}
    \item \nameref{appendix:CRR_pseudo}
    \item \nameref{appendix:versatile_stable}
    \item \nameref{appendix:appendix_Smoothing_functions}
    \item \nameref{appendix:experiments}
    \item \nameref{appendix:additional_results}
\end{enumerate}

\section{Missing Proofs}\label{appendix:missing_proofs}

We provide the proofs missing from the main text.
\appendixProofText
\section{CCR Pseudocode}\label{appendix:CRR_pseudo}
Algorithm~\ref{alg:CCR_algorithm} describes the operation of CCR rules.

\begin{algorithm}[ht]
\caption{Consensus-Conditioned Rule ($\mathrm{CCR}_{[\mathrm{CSF}, x]}$)}
\label{alg:CCR_algorithm}

\raggedright
\textbf{Input}: $\mathcal{T}=(A, E), \mathrm{CCR}=\bigl(\mathrm{CSF}, x\bigr)$

\begin{algorithmic}[1] %[1] enables  numbers
\STATE $S \gets \emptyset$
\STATE $score \gets 0$

\FOR{$t = 1$ to $|E|$}
    \STATE Let $E[t] = \langle e_1, ...,e_t \rangle$
    \FOR{each paragraph $p$ in $P(E[t])$}
        \STATE Compute: $score \gets \text{CSF}(p, E[t])$
        \IF{$score \ge x  \textbf{ and}  p \notin S$}
            \STATE \textbf{Attach:} $S \gets S \cup \{p\}$
        \ELSIF{$score < x  \textbf{ and }  p \in S$}
            \STATE \textbf{Detach:} $S \gets S \setminus \{p\}$
        \ENDIF
    \ENDFOR
\ENDFOR
\end{algorithmic}
\textbf{Output}: $S$ (final set of paragraphs selected) 
\end{algorithm}

% Figure 6 - Exponential smoothing function

\begin{figure*}[ht]
  \centering
  \begin{subfigure}[t]{0.3\textwidth}
    \includegraphics[width=\linewidth]{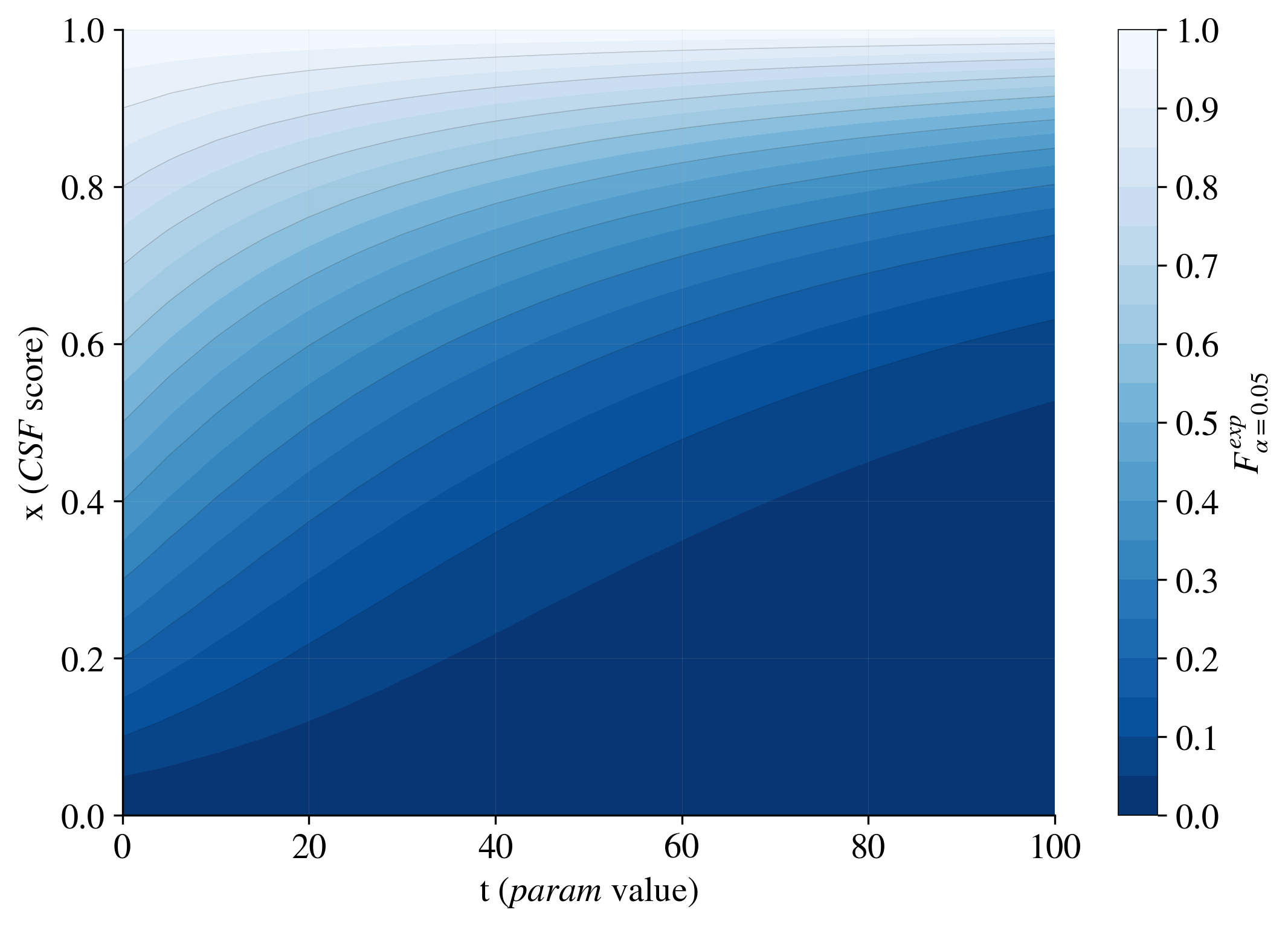}
    \caption{$\alpha=0.05$}
    \label{fig:0.05}
    \nicespace
  \end{subfigure}
  \hfill
  \begin{subfigure}[t]{0.3\textwidth}
    \includegraphics[width=\linewidth]{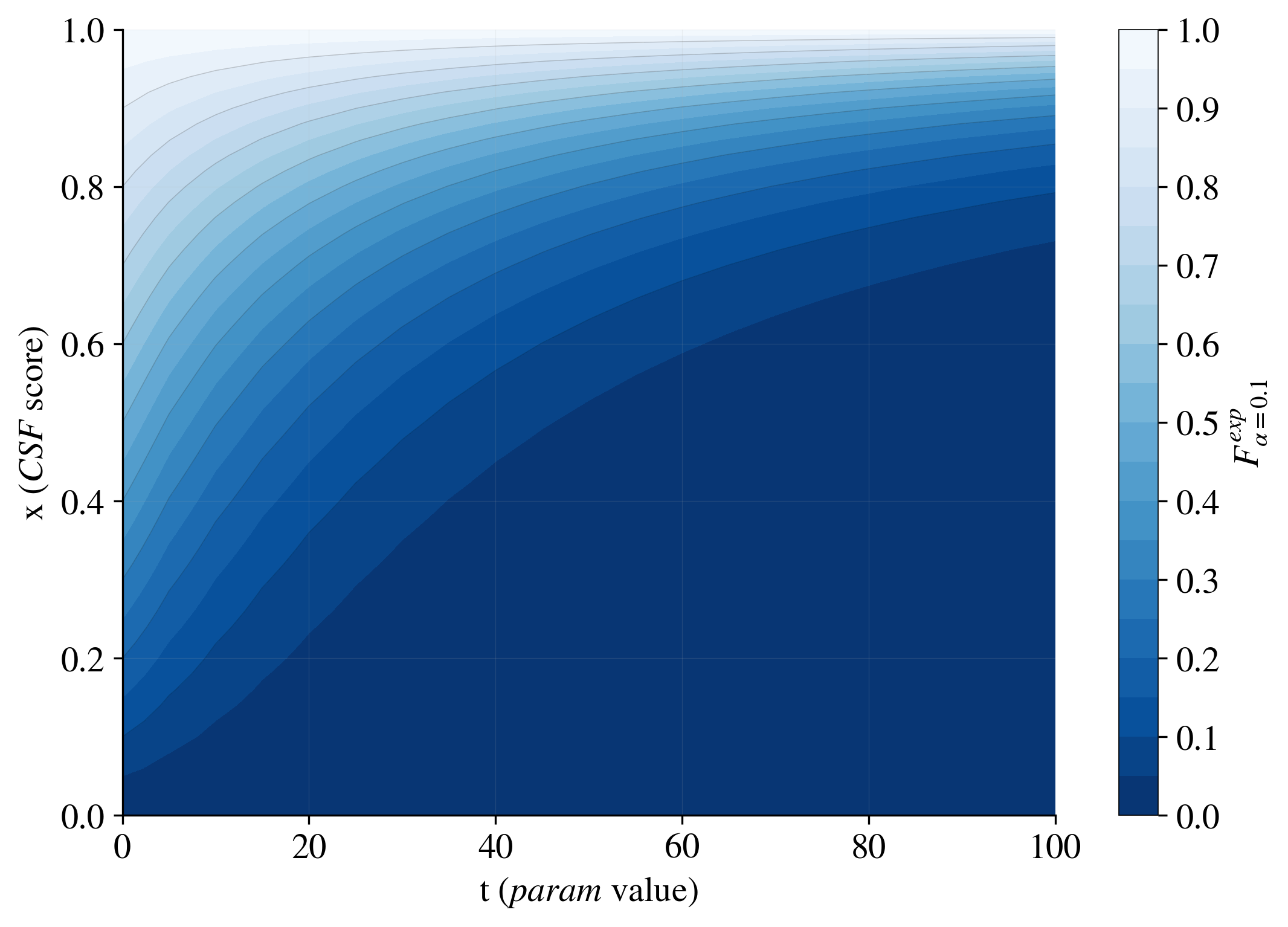}
    \caption{$\alpha=0.1$}
    \label{fig:0.1}
    \nicespace
  \end{subfigure}
  \hfill
  \begin{subfigure}[t]{0.3\textwidth}
    \includegraphics[width=\linewidth]{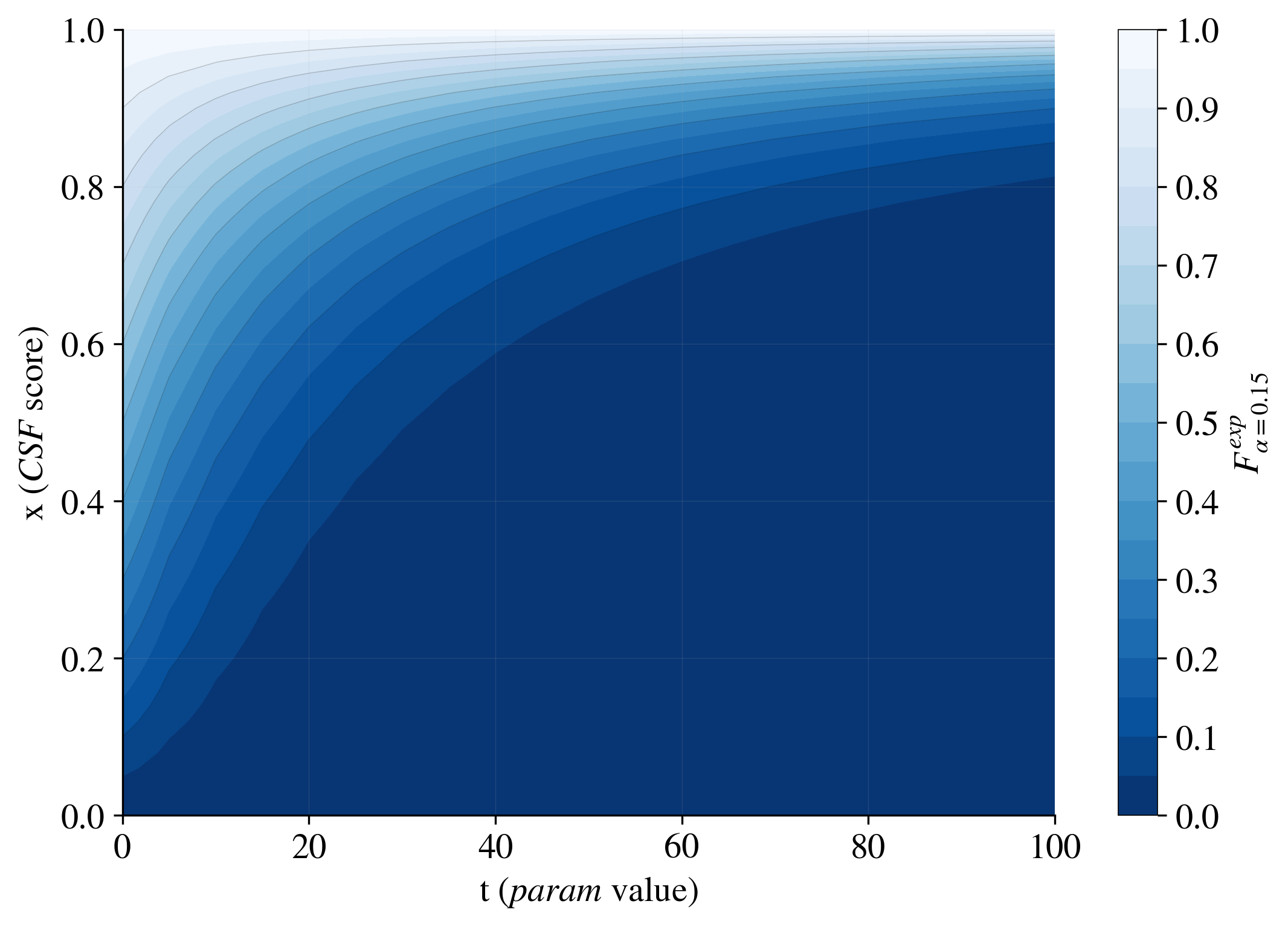}
    \caption{$\alpha=0.15$}
    \label{fig:0.15}
    \nicespace
  \end{subfigure}
    \caption{Exponential smoothing function ($F^{exp}_{\alpha}$) with $\alpha=0.05$ (left), $\alpha=0.1$ (middle), and $\alpha=0.15$ (right). 
    Larger values of $\alpha$ produce sharper transitions, particularly for higher $t$. This reflects the \textbf{exponential decay behavior} where higher $\alpha$ yields steeper drops for a given value of $t$. }
    %Increasing the value of $\alpha$ leads to a sharper transition in the function's values, particularly for larger values of $t$. This reflects the \textbf{exponential decay behavior}, where higher values of $\alpha$ lead to sharper drops in the function for a given value of the dynamic property $param$.}
    \label{fig:exp_alpha}
    \nicespace
\end{figure*}

\section{A Versatile, Stable, Non-static CCR}\label{appendix:versatile_stable}

We provide an example of a non-static CCR that is both versatile and stable.

\begin{example}
%
% Let $CSF$ be a (non-static) CSF  with $CSF(p, E)$ being computed as follows.
Consider a non-static consensus scoring function $\mathrm{CSF}$, computed as follows: it traverses the event list $E$ in order, but first looks for an abstention vote on some predefined paragraph $p^*$ -- i.e., if there is an event of the form $e = (a, p^*, 0)$ in $E$.
%-- 
Let $E^*$ denote the prefix of $E$ up to this special event (not including). Then, for each paragraph $p$, it assigns a score $1$ if $E^*$ contains \textbf{any} vote on $p$; otherwise, it assigns a score of $0$.
Indeed, the $\mathrm{CCR}_{[\mathrm{CSF}, 0.5]}$ rule is both versatile and stable:
\begin{itemize}

\item
\textbf{Versatility}:
For some set $P' = \{p_1,...,p_n\} \subseteq P$, construct an event list $E' = \langle (a,p_1,+1), ..., (a,p_n,+1), (a,p^*,0) \rangle$ and note that $\mathrm{CCR}_{[\mathrm{CSF}, 0.5]}(E') = P'$.

\item
\textbf{Stability}:
%
% For any event list $E' = \langle (a, p^*, 0)$ to any event list $E$ results in $CCR_{[CSF, 0.5]}(E + E') = CCR_{[CSF, 0.5]}(E + E' + E'')$ for any event list $E''$.
Consider any event list $E$, and let $E' = \langle (a, p^*, 0)\rangle$. Once this abstention event is appended to $E$, the solution remains unchanged under any further events. Formally, for any additional sequence of events $E''$,  $\mathrm{CCR}_{[\mathrm{CSF}, 0.5]}(E + E') = \mathrm{CCR}_{[\mathrm{CSF}, 0.5]}(E + E' + E'')$.

\end{itemize}
\end{example}

\section{Smoothing Functions}\label{appendix:appendix_Smoothing_functions}

\subsection{Smoothing Properties}\label{appendix:appendix_Smoothing_functions_properties}

To formally demonstrate that the examined smoothing functions fulfill the required properties (Section~\ref{section:Smooth_def}), we proceed as follows:

\begin{itemize}
    \item \textbf{Monotonic decrease with $t$}: 
    As $t$ increases, the $F_{\alpha}$ scores decrease monotonically. Formally, for all $\alpha \in [0,1]$, 
    $x\in [0,1]$ and any $t_1, t_2 \in \mathbb{N}$: if $t_2 \geq t_1$, then $F_{\alpha}(x, t_2) \leq F_{\alpha}(x, t_1)$.
    % \item \textbf{Monotonically Decreases with $t$}: 
    % As $t$ increases, the $F_{\alpha}$ - scores decrease monotonically. Formally, $\forall \alpha \in [0,1]$, $\forall t_1, t_2 \in \mathbb{N}$ and $\forall x\in [0,1]$: if $t_2 \geq t_1$ then $F_{\alpha}(x, t_2) \leq F_{\alpha}(x, t_1)$. 
    %
    \item \textbf{Monotonic decrease with $\alpha$}: The smoothing parameter $\alpha$ governs the rate of
    decay; a larger $\alpha$ results in a more hardened decrease of the function gradient. Formally, for all $x\in [0,1]$, $t \in \mathbb{N}$ and any $\alpha_1, \alpha_2 \in [0,1]$: if $\alpha_2 \geq \alpha_1$, then $F_{\alpha_2}(x, t) \leq F_{\alpha_1}(x, t)$.
    % \item \textbf{Monotonic decreases with $\alpha$}: The smoothing parameter $\alpha$ governs the rate of
    % decay; a larger $\alpha$ results in a more hardened decrease of the function gradient. Formally, $\forall \alpha_1, \alpha_2 \in [0,1]$, $\forall x\in [0,1]$, $\forall t \in \mathbb{N}$: if $\alpha_2 \geq \alpha_1$ then $F_{\alpha_2}(x, t) \leq F_{\alpha_1}(x, t)$.
    %
    \item \textbf{Monotonic increase with $x$}: 
    For a fixed $t$ and $\alpha$, the value of $F_\alpha(x, t)$ exhibits a positive correlation with the base score value $x$. Formally, given any $t \in \mathbb{N}$, $\alpha \in [0,1]$ and for all $x_1, x_2 \in [0,1]$: if $x_2 \geq x_1$, then $F_{\alpha}(x_2, t) \geq F_{\alpha}(x_1, t)$.
    %  \item \textbf{Monotonic increase with $x$}: 
    % The function F exhibits a positive correlation with the score value $x$. Formally, $\forall t \in \mathbb{N}$, $\forall x_1, x_2 \in [0,1]$, $\forall \alpha \in [0,1]$: if $x_2 \geq x_1$ then $F_{\alpha}(x_2, t) \geq F_{\alpha}(x_1, t)$. 
    %
    \item \textbf{Amendability preservation}: $F_{\alpha}(x,t)$ never outputs zero for a nonzero input score $x$, preserving the possibility of future amendments. Formally, if $x \neq 0$, then $F_{\alpha}(x, t) \neq 0$ for all $t \in \mathbb{N}$ and $\alpha \in [0,1]$.
\end{itemize}

\subsection{Exponential Smoothing Family}\label{appendix:appendix_Smoothing_functions_behavior}

For the examined exponential family of smoothing functions, we prove that they indeed satisfy these properties:
\[
\text{\textbf{Exponential}: } F^{exp}_{\alpha}(x, t) = x \cdot e^{-t \cdot \alpha \cdot (1 - x)}
\]
\begin{itemize}
    \item Monotonic decrease with $t$: The derivative of $F^{exp}_{\alpha}$ with respect to $t$ is
    $$
    \frac{\partial F^{exp}_{\alpha}}{\partial t} = -x \cdot \alpha \cdot (1 - x) \cdot e^{-t \cdot \alpha \cdot (1 - x)},
    $$
    which is a negative, since for any $x \in [0, 1]$ and $\alpha \in [0, 1]$ -- $e^{-t \cdot \alpha \cdot (1-x)}>0$, $(1 - x) \ge 0$ but $-x \cdot \alpha \le 0$.
    %$x \in [0, 1]$, $\alpha \in [0, 1]$, and $e^{-t \cdot \alpha \cdot (1 - x)} > 0$, we have $\frac{\partial F^{exp}_{\alpha}}{\partial t} \leq 0$. 
    Hence, $F^{exp}_{\alpha}$ decreases monotonically with $t$.
    \item Monotonic decrease with $\alpha$: The derivative of $F^{exp}_{\alpha}$ with respect to $\alpha$ is
    $$
    \frac{\partial F^{exp}_{\alpha}}{\partial \alpha} = -x \cdot t \cdot (1 - x) \cdot e^{-t \cdot \alpha \cdot (1 - x)},
    $$
    which is a negative, as for any $x \in [0, 1]$ and $t \geq 0$, the term $\frac{\partial F^{exp}_{\alpha}}{\partial \alpha} \leq 0$ , ensuring that $F^{exp}_{\alpha}$ decreases monotonically with $\alpha$.
    \item Monotonic increase with $x$: The partial derivative of $F^{exp}_{\alpha}$ with respect to $x$ is
    $$
    \frac{\partial F^{exp}_{\alpha}}{\partial x} = e^{-t \cdot \alpha \cdot (1 - x)} \left[ 1 + t \cdot \alpha \cdot x \right],
    $$
    which is nonnegative, since  $e^{-t \cdot \alpha \cdot (1 - x)} > 0$  and $1 + t \cdot \alpha \cdot x \geq 0$. Thus, $F^{exp}_{\alpha}(x,t)$ increases monotonically with $x$.
    %As $e^{-t \cdot \alpha \cdot (1 - x)} > 0$  and $1 + t \cdot \alpha \cdot x \geq 0$ , we have \( \frac{\partial F^{exp}_{\alpha}}{\partial x} \geq 0 \), meaning that $F^{exp}_{\alpha}$ increases monotonically with $x$.
    %
    \item Amendability Preservation: For $x \neq 0$, we observe that as $t \to \infty$, $F^{exp}_{\alpha}(x, t)$  approaches zero asymptotically but never reaching. Therefore,  $F^{exp}_{\alpha}(x, t) \neq 0$  for any finite $t$, preserving the possibility for future amendments.
    % Moreover, when $x = 1$, then $F^{exp}_{\alpha} = 1$, which ensures that the highest-scoring paragraphs retain their influence. 
\end{itemize}

Note that $F^{exp}_{\alpha}$ retains the anticipated boundary behavior: scores near 1 remain influential, while scores near 0 remain negligible, ensuring that strong and weak paragraphs preserve their standing.

Figure~\ref{fig:exp_alpha} illustrates the behavior of the exponential smoothing function $F^{exp}_{\alpha}(x, t)$, showing how different values of $\alpha$ affect the functional response across $x$ (the static-CSF in our use) and $t$ (the $param$ value).

\section{Experiments: Additional Details}\label{appendix:experiments}

\begin{figure*}[t]
    \centering
    \begin{subfigure}[t]{0.35\textwidth}
        \centering
        \includegraphics[width=\textwidth, height=0.3\textheight, keepaspectratio]{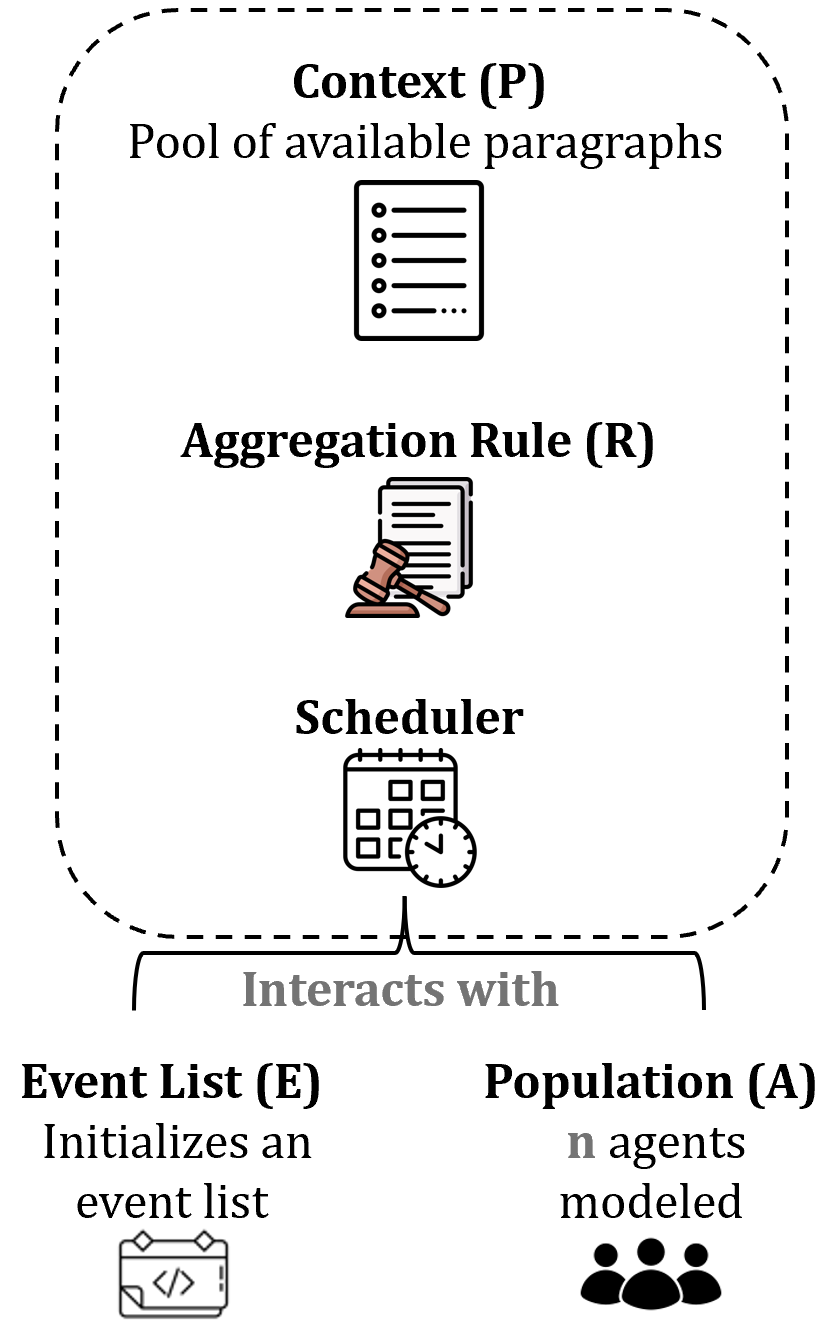}
        \caption{\centering Simulation components defining a system instance.}
        \nicespace
\label{fig:plaformA}
    \end{subfigure}
    \hspace{0.06\textwidth}
    \begin{subfigure}[t]{0.55\textwidth}
        \centering
        \includegraphics[width=\textwidth, height=0.3\textheight, keepaspectratio]{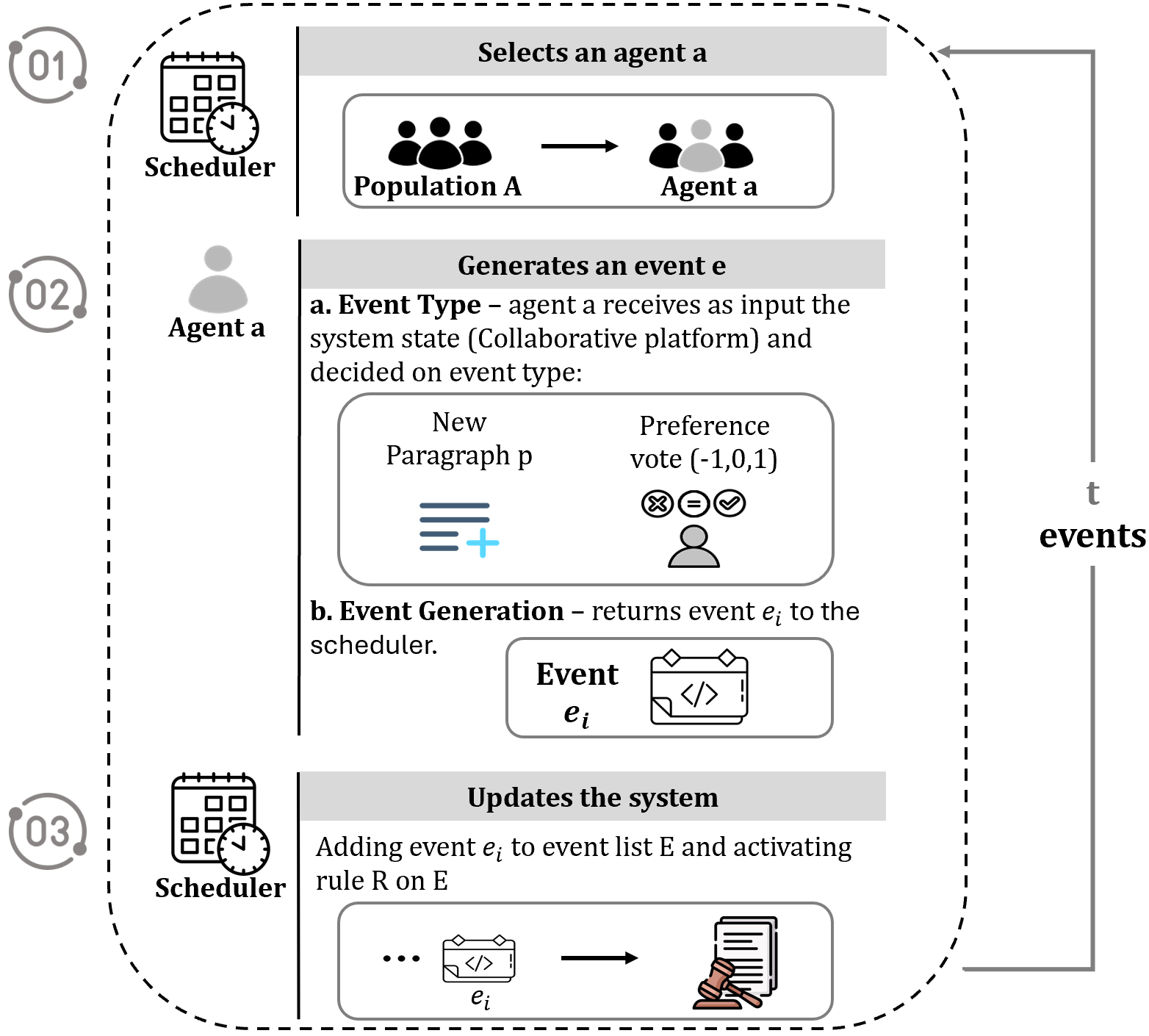}
        \caption{\centering Illustration of the simulation design described in Section~\ref{section:experimental_design}. The scheduler interacts with the agent population $A$ of $n$ agents in thematic context $P$, populating an event list $E$ with $t$ events.}
\nicespace
          \label{fig:simulation_design}
    \end{subfigure}
    \caption{Overview of the simulation setup.}
    \label{fig:simulation_overview}
\nicespace
\end{figure*}

\subsection{Simulation Components Illustrations}\label{appendix:Simulations Design}

Figure~\ref{fig:plaformA} illustrates the abstract simulation components, and Figure~\ref{fig:simulation_design} illustrates the simulation design.

\subsection{Datasets}\label{appendix:datasets}

The multidimensional structure enables proportional sampling of agent profiles across realistic demographic categories. These distributions are derived from national census estimates. 

\begin{description}
    \item[Demographic UN dataset:] 
    We use the United Nations dataset “Population 15 years of age and over, by educational attainment, age and sex” (UNdata)~\cite{UNdata2022}, filtered to Israel (2022) as the most recent data. This dataset provides absolute population counts (in thousands) for each combination of sex, age group, and education levels (ISCED classification). Age groups are reported in 5-year intervals, and sex is binary (male and female). Population counts were normalized to proportions, forming a weighted sampling distribution over joint demographic profiles (may not sum exactly due to rounding).
    \item [Synthetic Climate Action Proposal Dataset:] This dataset spans ten key domains relevant to our use-case topic -- urban climate governance (Building, Transport, Energy, Waste, Water, Health, Business, Natural Environment, Land Use, and Natural Hazards). Within each domain, we generated a series of proposals across a continuous activism sentiment scale from $0.0$ to $1.0$, divided into five sentiment categories:
    \begin{itemize}[label=-]
        \item \textbf{Active Resistance (0.00--0.20):} Proposals that oppose climate action or dismantle existing measures.
        \item \textbf{Minimal Acknowledgment (0.20--0.40):} Proposals that recognize climate concerns without enforcing meaningful change.
        \item \textbf{Balanced Approach (0.40--0.60):} Proposals that weigh climate action alongside economic or practical constraints.
        \item \textbf{Supportive Measures (0.60--0.80):} Proposals that endorse proactive but feasible climate policies.
        \item \textbf{Proactive Action (0.80--1.00):} Proposals that prioritize ambitious climate-first transformations.
    \end{itemize}
    Thus, each proposal entry includes its domain, sentiment score, textual description, and a rationale explaining its intent or impact. We generated 21 proposals per domain (one at each 0.05 sentiment increment), for a total of just over 200 proposals, and combined them into a single JSON dataset.
    The proposals are grouped by sentiment category to support dynamic \emph{few-shot prompting}. At runtime, for each API call, three different subtopic proposals matching the agent's sentiment category are sampled to guide contextually relevant decision-making. 
    
    \vspace{0.5\baselineskip}

    \textbf{Examples: Active Resistance}
    \begin{itemize}
        \item \textbf{[Building]} Mandate fixed temperature settings in all buildings year-round without adjustment for occupancy or seasons. \\
        \textit{Reasoning:} Enforces energy-wasteful practices.
        \item \textbf{[Energy]} Require all municipal energy contracts to prioritize the lowest cost regardless of source, with explicit exclusion of renewable premium payments. \\
        \textit{Reasoning:} Systematically disadvantages clean energy options.
    \end{itemize}

    \textbf{Examples: Balanced Approach}
    \begin{itemize}
        \item \textbf{[Building]} Consider both energy efficiency and traditional cost metrics equally in building decisions. \\
        \textit{Reasoning:} Balanced consideration without preference.
        \item \textbf{[Energy]} Create a public dashboard showing real-time municipal renewable energy generation. \\
        \textit{Reasoning:} Transparency promoting voluntary action.
    \end{itemize}

\end{description}

\subsection{Prompts}\label{appendix:prompts}

For completeness, we detail the prompts used for LLM-based agent simulations. These prompts are rendered using LangChain's ChatPromptTemplate to maintain multi-turn interactions:

\begin{tcolorbox}[title=Chain Template (ChatPromptTemplate), width=\linewidth, breakable]
\rmfamily\small
\textbf{chat\_prompt} = ChatPromptTemplate.from\_messages([\\
{SystemMessage(content=\textcolor{Cyan}{\bb{\textbf{\textit{system\_prompt}}}}),}\\
{HumanMessage(content=\textcolor{Cyan}{\bb{\textbf{\textit{decision\_prompt}}}})}\\
{])}
\end{tcolorbox}

\subsubsection*{System Prompt Design}\label{section:system_prompt}

The system prompt introduces each agent, specifying its role in a collaborative policy-writing task. 
The template includes placeholder variables such as \emph{agent id}, \emph{topic}, \emph{profile}, \emph{topic position}, and its \emph{description}. These slots are populated at runtime to customize the prompt for each agent. The system prompt thereby positions the agent within a specific background context and defines their activism position on the topic. Our system prompt follows the following template:

\begin{tcolorbox}[title=System Prompt Template, width=\linewidth]
\small
\raggedright
%\rmfamily
%
You are Agent \textcolor{cyan}{\bb{\textbf{\textit{agent\_id}}}}, a participant in a collaborative constitution writing system focused on \textcolor{cyan}{\bb{\textbf{\textit{topic}}}}. \\ [2pt]
You engage dynamically with an evolving draft document, consisting of community proposals. The document is a list of action proposals that your community is considering adopting as policy. \\[2pt]
\textbf{YOUR PROFILE}: \textcolor{cyan}{\bb{\textbf{\textit{profile}}}}. \\
\textbf{YOUR POSITION ON} \textcolor{cyan}{\bb{\textbf{\textit{topic}}}}: \textcolor{cyan}{\bb{\textbf{\textit{topic\_position}}}}(\textcolor{cyan}{\bb{\textbf{\textit{position\_category}}}}). \\[2pt]
Accordingly, you have the following orientation: \textcolor{cyan}{\bb{\textbf{\textit{description}}}}. \\
As a participant, you can decide to vote on existing proposals (upvote, downvote, abstain) or suggest a new proposal. \\
Act consistently with your profile. Think step-by-step; align your reasoning explicitly with your stated orientation.
\end{tcolorbox}

% \begin{quote}[\textbf{System Prompt}]
% \rmfamily\small

% You are Agent \bb{\textbf{\textit{agent\_id}}}, a participant in a collaborative constitution writing system focused on \bb{\textbf{\textit{topic}}}. \\
% You engage dynamically with an evolving draft document, consisting of community proposals. \\
% The document is a list of action proposals that your community is considering
% adopting as policy. \\
% YOUR PROFILE: \bb{\textbf{\textit{profile}}} \\
% YOUR POSITION ON \bb{\textbf{\textit{topic}}}: \bb{\textbf{\textit{topic position}}} (\bb{\textbf{\textit{position category}}}). \\
% Accordingly, you have the following orientation: \bb{\textbf{\textit{description}}} \\
% As a participant, you can decide on voting on existing proposals (upvote, downvote, abstain) or suggest a new proposal. \\
% Act consistently with your profile. Think step-by-step; align your reasoning explicitly with your stated orientation.
% \end{quote}

\subsubsection*{Decision Prompt Design}\label{section:decision_prompt}
The decision prompt presents the document's current state from the agent's point of view, the stance of the system, its current preferences, and the current document composition. It guides an agent using a structured \emph{Chain-of-Thought (CoT)} approach, facilitating stepwise reasoning for two primary sequential decisions:
\begin{enumerate}
    \item Whether to vote or propose.
    \item If voting, selecting the target paragraph and preference; or, if proposing, generating a new diverse proposal.
\end{enumerate} 
Variable placeholders include \emph{current state}, \emph{action hint}, and dynamically integrated  \emph{examples}. Our decision prompt follows the following template:
\begin{tcolorbox}[breakable, title=Decision Prompt Template, width=\linewidth]
\small
\raggedright
%\rmfamily
%
\textbf{SYSTEM STATE:} \textcolor{cyan}{\bb{\textbf{\textit{system\_state\_section}}}}.\\ [2pt]
\textbf{TASK:} Based on the system state and your profile, you must decide your next action through the following reasoning steps. \\
\textcolor{cyan}{\bb{\textbf{\textit{action hint}}}}. Be very concise when addressing each step. \\ [2pt]
\textbf{Step 1: Situation Analysis} \\ [2pt]
- Analyze current proposals, subtopics, their vote counts, and your own stance. \\
- Note which paragraphs are included in the document and which are not. \\
- Identify over-represented (repeated ideas), under-represented (missing themes), or unbalanced (extreme/vague) topics. \\ [2pt]
\textbf{Step 2: Position Alignment} \\ [2pt]
- Evaluate if the included paragraphs reflect your profile and position. \\
- Identify proposals/content you support or oppose, and whether they align with your orientation. \\
- Identify critical gaps that would make the document incomplete or unbalanced from your perspective. \\ [2pt]
\textbf{Step 3: Decision Evaluation} \\ [2pt]
- Decide whether proposing or voting best influences the document at this point. \\
- If your support for a proposal has weakened (own vote not "?"), Choose ABSTAIN to reflect neutrality. \\
- Determine the optimal strategic action: propose a new paragraph/vote on an existing one. \\ [2pt]
\textbf{Step 4: Action Formulation} \\ [2pt]
\textit{If PROPOSE:} Formulate a new proposal that: \\
- Addresses a sub-topic important to your background \\
- Is practical, realistic, diverse (max 20 words) \\
- Aligns with these examples of your position: \textcolor{cyan}{\bb{\textbf{\textit{self.\_format\_examples\_for\_prompt()}}}} \\
\textit{If VOTE:} Choose a paragraph and indicate clearly if you UPVOTE (agree), DOWNVOTE (disagree), or ABSTAIN (change previous vote to undecided). \\ [4pt]
\textbf{Voting Rules (Mandatory):} \\ [2pt]
- You may only ABSTAIN on a paragraph where your current vote is +1 or -1. \\
- If your current vote is '?', choose only UPVOTE or DOWNVOTE. \\
- Do NOT repeat your current vote.  \\ [4pt]
\textbf{Respond with the following structure (Mandatory):}  \\ [2pt]
\textbf{DECISION:} PROPOSE / VOTE \\
\textbf{PARAGRAPH ID:} [If VOTE: Paragraph ID, If PROPOSE: Next available ID] \\
\textbf{ACTION DETAILS:} [If VOTE: Proposal text, If PROPOSE: New proposal text] \\
\textbf{VOTE:} [If VOTE: UPVOTE / DOWNVOTE / ABSTAIN, If PROPOSE: UPVOTE only] \\
\textbf{REASONING:} [Brief justification, max 20 words, aligned with profile, position, and state].
\end{tcolorbox}

\subsection{Scheduling Processes Examples}\label{appendix:simulations_schedulers}
\begin{example}[Unstructered Population Scheduler]\label{ex:unstractured_scheduler} 
Here is a possible simulation run with an unstructured population (see Section~\ref{section:unstructured_population}). We define the system with context $P = \{p_1,p_2,p_3\}$ and agent population $A = \{a_1,a_2,a_3\}$ and number of events is three ($|E| = 3$). In the first iteration, the scheduler randomly selects agent $a_2$, who proposes a new paragraph $p_1$ since no paragraphs exist yet, leading to the event $e_1 = (a_2, p_1, +1)$. In the second iteration, agent $a_1$ performs a preference voting event on $p_1$ and votes in favor, adding event $e_2 = (a_1, p_1, +1)$. During the third iteration, $a_3$ is selected; however, since they have no prior votes, they are unable to perform the drawn event type of preferred vote. Hence, they propose $p_2$, resulting in $e_3 = (a_3, p_2, +1)$. Once the list reaches the required length, the process stops, and the resulting event list is: $E = \langle (a_2, p_1, +1), (a_1, p_1, +1), (a_3, p_2, +1) \rangle$.
\end{example}
\noindent\textbf{Detailed Instance and Event Generation}
\\ [2pt]
\noindent We provide supplementary details for the simulation instance presented in Example~\ref{example:llm_schedule} (a full log is available in the project repository). The instance involves 20 LLM-based agents distributed across five sentiment categories toward climate change policy: six proactive actions, five supportive measures, four minimal acknowledgments, four active resistances, and one balanced approach. Events are scheduled under the rule $\mathrm{CCR}_{[\mathrm{APS}, 0.7]}$. Here, we focus on the generation of event $ e_{38} = (a_2, p_2,+1)$ (proposal event of $p_2$). Table~\ref{tab:example_event_generation_tally} presents the state of the system prior to event $e_{38}$.

\begin{table}[!ht]
    \centering
    \footnotesize
    \caption{Tally matrix for Example~\ref{example:llm_schedule} after the initial $37$ events ($\mathrm{tally}(E[37])$). Under the rule $(\mathrm{APS}, 0.7)$, paragraph $p_1$ is currently included in the document (as $\mathrm{APS}(p_1, E[37]) = \nicefrac{12}{17} > 0.7$).}
    \setlength{\tabcolsep}{3pt} % reduce column padding
    \begin{tabular}{c@{\hspace{1mm}}p{4.6cm}@{\hspace{1mm}}c@{\hspace{1mm}}c}
    \toprule
    \multicolumn{1}{c}{\textbf{Paragraph}} & \multicolumn{1}{c}{\textbf{Text}} & \multicolumn{1}{c}{\textbf{+ Votes}} & \multicolumn{1}{c}{\textbf{- Votes}} \\
    \midrule
    \multirow{2}{*}{$p_1$} & Implement community gardens in urban areas to promote local food production and biodiversity. & \multirow{2}{*}{$12$} & \multirow{2}{*}{$5$} \\
    \bottomrule
    \end{tabular}
    \label{tab:example_event_generation_tally}
\end{table}

Next, we show the prompts that led to the generation of event $e_{38}$. Below is the full system prompt issued to the selected agent $a_2$:

\begin{tcolorbox}[title=System Prompt Example, width=\linewidth]
\small
\raggedright
You are Agent \textcolor{cyan}{\textbf{2}}, a participant in a collaborative constitution writing system focused on \textcolor{cyan}{\textbf{Climate change policy}}. \\[2pt]
You engage dynamically with an evolving draft document, consisting of community proposals. The document is a list of action proposals that your community is considering adopting as policy. \\[2pt]
\textbf{YOUR PROFILE}: \textcolor{cyan}{\textbf{You are a 30–34 years old female with ISCED 2011, level 3 – upper secondary education}}. \\[2pt]
\textbf{YOUR POSITION ON \textcolor{cyan}{\textbf{Climate change policy}}}: \textcolor{cyan}{\textbf{0.05 (active resistance})}. \\[2pt]
Accordingly, you have the following orientation: \textcolor{cyan}{\textbf{complete opposition to in-favor action}}. \\[2pt]
As a participant, you can decide to vote on existing proposals (upvote, downvote, abstain) or suggest a new proposal. \\
Act consistently with your profile. Think step-by-step; align your reasoning explicitly with your stated orientation.
\end{tcolorbox}

The corresponding decision prompt excerpt is presented below.

\begin{tcolorbox}[breakable, title=Decision Prompt - Placeholders and Response, width=\linewidth]
\small
\raggedright
\textbf{SYSTEM STATE:} \textcolor{cyan}{\textbf{\{'paragraph\_id': 1, 'text': 'Implement community gardens in urban areas to promote local food production and biodiversity.', 'votes\_plus': 12, 'votes\_minus': 5, 'own\_vote': '-1', 'In document': 'yes'\}}}. \\[6pt]
\textbf{TASK:} Based on the system state and your profile, you must decide your next action through the following reasoning steps. \\
\textcolor{cyan}{\textbf{You can either PROPOSE a new action or VOTE on an existing proposal}}. Be very concise when addressing each step. \\[6pt]
\textbf{Step 4: Action Formulation} \\[2pt]
\textit{If PROPOSE:} Formulate a new proposal that: \\
- Addresses a sub-topic important to your background \\
- Is practical, realistic, diverse (max 20 words) \\
- Aligns with these examples of your position: \\
\textcolor{cyan}{\textbf{1. [Health] Consider both traditional health metrics and climate-related health impacts in program planning. (Reasoning: Balances conventional and climate-related health considerations equally).}} \\
\textcolor{cyan}{\textbf{2. [Building] Maintain current industry-standard energy codes for all new municipal buildings. (Reasoning: Meets minimum requirements without exceeding them).}} \\
\textcolor{cyan}{\textbf{3. [Land Use] Consider both compact and traditional development patterns equally in zoning decisions. (Reasoning: Balances different development approaches without preference).}} \\[4pt]
\textit{If VOTE:} Choose a paragraph and indicate clearly if you UPVOTE (agree), DOWNVOTE (disagree), or ABSTAIN (change previous vote to undecided). \\[6pt]
\textbf{Response:} \\[4pt]
\textbf{DECISION:} \textcolor{cyan}{\textbf{PROPOSE}} \\
\textbf{PARAGRAPH ID:} \textcolor{cyan}{\textbf{2}} \\
\textbf{ACTION DETAILS:} \textcolor{cyan}{\textbf{Maintain current agricultural practices without promoting new sustainability initiatives}} \\
\textbf{VOTE:} \textcolor{cyan}{\textbf{UPVOTE}} \\
\textbf{REASONING:} \textcolor{cyan}{\textbf{This aligns with my position against new environmental actions and supports traditional practices}.}
\end{tcolorbox}

%%%%%%%%%%%%%%%%%%%%%%%%%%%%%%%%%%%%%%%%%%%%%%%%%%%%%%%%%%%%%%%%%%%%%%%%

%
\section{Additional Results}\label{appendix:additional_results}
\subsection{Solution Size Analysis}\label{appendix:solution_size_analysis}
Table~\ref{tab:solution_size_summary} summarizes the solution sizes obtained in different agent models and rule families at event 300. Notably, the uniform semantic distribution in Euclidean space consistently resulted in much larger solution sizes, particularly under the Static and Smooth conditions. The Smooth condition also exhibited greater variability in solution size, indicating a higher sensitivity to agent models compared to the Static and Harsh conditions.

\begin{table}[th]
    \centering
    \caption{Average solution size $\pm$ std.\ deviation for each agent model and rule family (54 rules, event 300).}
    \label{tab:solution_size_summary}
    \begin{tabular}{l|ccc}
    \toprule
    \multirow{2}{*}{\textbf{Agent Model}} & \multicolumn{3}{c}{\textbf{Rule Family}} \\
     & Static  & Harsh  & Smooth  \\
    \midrule
    Normal       & $5.33 \pm 0.34$ & $3.78 \pm 1.79$ & $1.35 \pm 0.97$ \\
    Two-Peaks    & $53.70 \pm 20.38$ & $4.09 \pm 4.00$ & $13.81 \pm 10.05$ \\
    Unstructured & $54.96 \pm 27.74$ & $5.87 \pm 7.99$ & $25.36 \pm 27.57$ \\
    Uniform      & $\boldsymbol{93.40 \pm 45.42}$ & $\boldsymbol{6.46 \pm 8.73}$ & $\boldsymbol{32.84 \pm 34.70}$ \\
    \bottomrule
    \end{tabular}
\end{table}

\subsection{Extended Pareto Analysis}\label{appendix:extended_pareto_analysis}

\begin{figure}[bht]
    \centering
    \includegraphics[width=0.45\textwidth, height=0.45\textheight, keepaspectratio]  {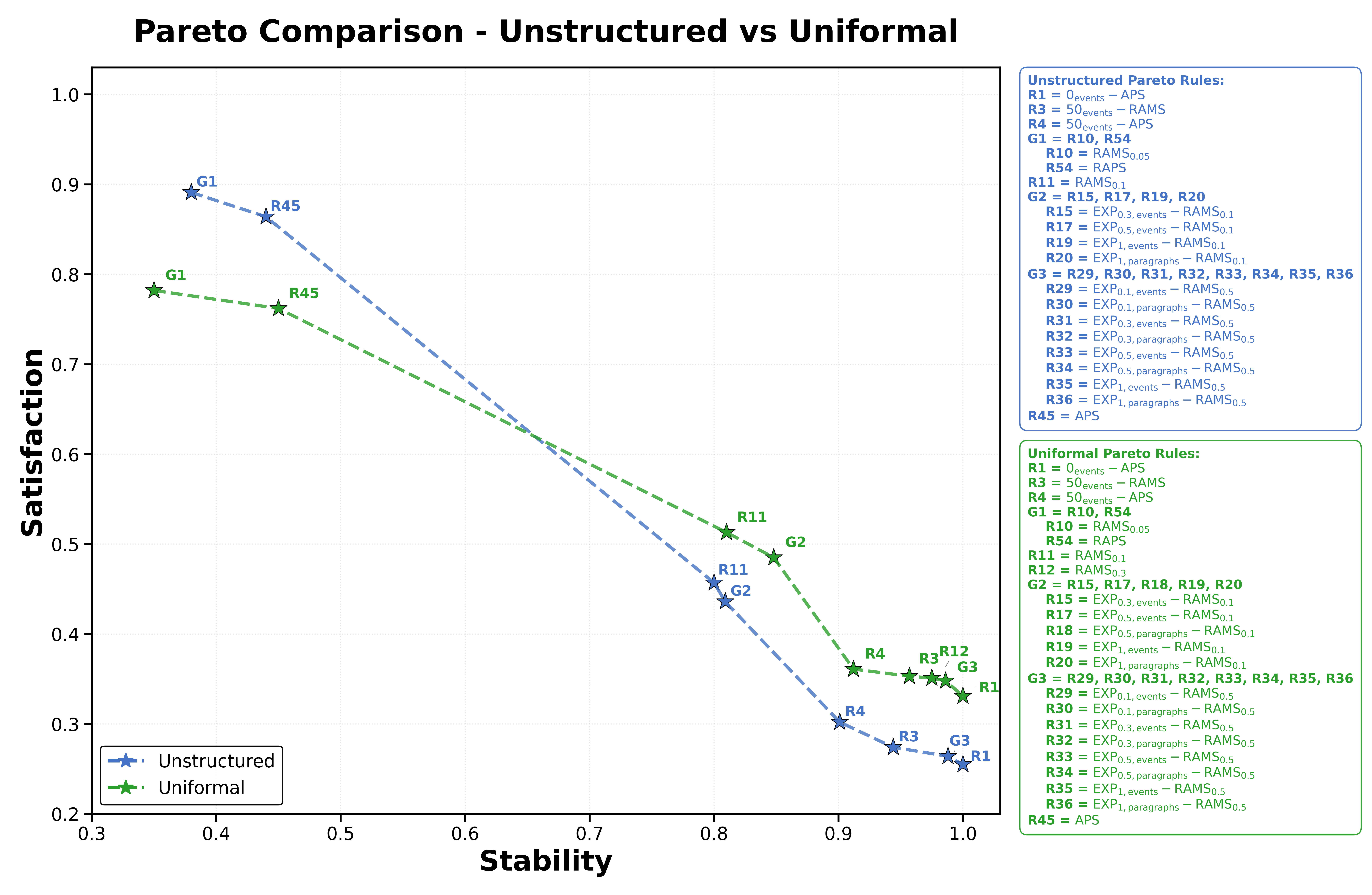}
    \caption{Comparison of the satisfaction-stability Pareto trade-offs achieved under unstructured vs. uniform agent populations (instances with 20 agents and 300 events). As expected, the unstructured (random) population approximates a uniform sentiment distribution, resulting in a Pareto-optimal rule set very similar to the uniform population.}
    \label{fig:unstructured_uniformal_comparasion}
\nicespace
\end{figure}

\end{document}